\documentclass[11pt,letterpaper]{article}
\usepackage{enumerate}
\usepackage{etoolbox}

\usepackage{amsthm}
\usepackage{amsmath,latexsym,amssymb,verbatim,enumerate,graphicx}
\usepackage{hyperref}
\usepackage{color}
\usepackage{enumitem}
\usepackage{booktabs}
\RequirePackage[tracking=basictext, kerning=basictext, spacing=basictext]{microtype}
\RequirePackage{fixltx2e}
\SetTracking[spacing = {50*,0*,0*}]{encoding = *}{7}
\SetExtraKerning{encoding=*}{) = {25,},( = {,75}}
\RequirePackage{mathtools}

\newtheorem*{rep@theorem}{\rep@title}
\newcommand{\newreptheorem}[2]{%
\newenvironment{rep#1}[1]{%
 \def\rep@title{#2 \ref{##1} (restatement)}%
 \begin{rep@theorem}}%
{\end{rep@theorem}}}
\makeatother

\newtheorem{thm}{Theorem}
\newtheorem*{thm*}{Theorem}
\newtheorem{prop}[thm]{Proposition}
\newtheorem*{prop*}{Proposition}
\newtheorem{lem}[thm]{Lemma}
\newtheorem*{lem*}{Lemma}
\newtheorem{algorithm}[thm]{Algorithm}

\newtheorem*{fact*}{Fact}
\newtheorem{cor}[thm]{Corollary}
\newtheorem*{cor*}{Corollary}
\newtheorem{con}[thm]{Conjecture}

\newreptheorem{thm}{Theorem}
\newreptheorem{lem}{Lemma}
\newreptheorem{cor}{Corollary}
\newreptheorem{prop}{Proposition}

\def\ba#1\ea{\begin{align}#1\end{align}}
\def\ban#1\ean{\begin{align*}#1\end{align*}}
\newcommand{\subeq}[2]{\begin{subequations} \label{eq:#1} \begin{align}
      #2 \end{align} \end{subequations}}

\newcommand{\ot}{\otimes}
\newcommand{\be}{\begin{equation}}
\newcommand{\ee}{\end{equation}}

\def\BQP{{\sf{BQP}}}

\def\QMA{{\sf{QMA}}}

\def\NP{{\sf{NP}}}

\def\P{{\sf{P}}}

\def\SAT{{\sf{SAT}}}

\def\benum{\begin{enumerate}}
\def\eenum{\end{enumerate}}

\def\nn{\nonumber}

\def\squareforqed{\hbox{\rlap{$\sqcap$}$\sqcup$}}
\def\qed{\ifmmode\squareforqed\else{\unskip\nobreak\hfil
\penalty50\hskip1em\null\nobreak\hfil\squareforqed
\parfillskip=0pt\finalhyphendemerits=0\endgraf}\fi}
\def\endenv{\ifmmode\;\else{\unskip\nobreak\hfil
\penalty50\hskip1em\null\nobreak\hfil\;
\parfillskip=0pt\finalhyphendemerits=0\endgraf}\fi}

\newcommand{\bra}[1]{\langle #1|}
\newcommand{\ket}[1]{|#1\rangle}
\newcommand{\braket}[2]{\langle #1|#2\rangle}

\newcommand{\bbR}{\mathbb{R}}
\newcommand{\bbC}{\mathbb{C}}

\newcommand{\id}{\mathbb{I}}

\newcommand{\ben}{\begin{equation}}
\newcommand{\een}{\end{equation}}

\newcommand{\<}{\langle}
\renewcommand{\>}{\rangle}

\def\L{\left}
\def\R{\right}

\def\id{{\operatorname{id}}}

\DeclareMathOperator{\Cov}{Cov}

\DeclareMathOperator{\sgn}{sgn}
\DeclareMathOperator{\rank}{rank}

\DeclareMathOperator{\Span}{Span}
\DeclareMathOperator{\supp}{supp}
\DeclareMathOperator{\tr}{tr}

\def\be{\begin{equation}}
\def\ee{\end{equation}}
\def\ben{\begin{eqnarray}}
\def\een{\end{eqnarray}}
\def\ot{\otimes}

\def\bei{\begin{itemize}}
\def\eei{\end{itemize}}

\DeclareMathOperator*{\E}{\mathbb{E}}

\def\eps{\epsilon}

\def\cD{{\cal D}}

\def\cH{{\cal H}}
\def\cL{{\cal L}}

\def\cP{{\cal P}}

\def\cS{{\cal S}}
\def\cT{{\cal T}}

\mathchardef\ordinarycolon\mathcode`\:
\mathcode`\:=\string"8000
\def\vcentcolon{\mathrel{\mathop\ordinarycolon}}
\begingroup \catcode`\:=\active
  \lowercase{\endgroup
  \let :\vcentcolon
  }

\newcommand{\nc}{\newcommand}
\nc{\rnc}{\renewcommand} \nc{\beq}{\begin{equation}}
\nc{\eeq}{{\end{equation}}} \nc{\bea}{\begin{eqnarray}}
\nc{\eea}{\end{eqnarray}} \nc{\beqa}{\begin{eqnarray}}
\nc{\eeqa}{\end{eqnarray}} \nc{\lbar}[1]{\overline{#1}}
 \nc{\proj}[1]{|#1\rangle\!\langle #1 |} 
\nc{\avg}[1]{\langle#1\rangle}

\nc{\conv}{\operatorname{conv}}
\nc{\smfrac}[2]{\mbox{$\frac{#1}{#2}$}} \nc{\Tr}{\operatorname{Tr}}
\nc{\ox}{\otimes} \nc{\dg}{\dagger} \nc{\dn}{\downarrow}
\nc{\lmax}{\lambda_{\text{max}}}
\nc{\lmin}{\lambda_{\text{min}}}

\nc{\csupp}{{\operatorname{csupp}}}
\nc{\qsupp}{{\operatorname{qsupp}}} \nc{\var}{\operatorname{var}}
\nc{\la}{\leftarrow}\nc{\ra}{\rightarrow}
\nc{\rar}{\rightarrow} \nc{\lrar}{\longrightarrow}
\nc{\poly}{\operatorname{poly}}
\nc{\polylog}{\operatorname{polylog}} \nc{\Lip}{\operatorname{Lip}}
\nc{\Om}{\Omega}
\nc{\wt}[1]{\widetilde{#1}}

\def\>{\rangle}
\def\<{\langle}

\nc{\glneq}{{\raisebox{0.6ex}{$>$}  \hspace*{-1.8ex} \raisebox{-0.6ex}{$<$}}}
\nc{\gleq}{{\raisebox{0.6ex}{$\geq$}\hspace*{-1.8ex} \raisebox{-0.6ex}{$\leq$}}}

\nc{\vholder}[1]{\rule{0pt}{#1}}
\nc{\wh}[1]{\widehat{#1}}
\nc{\h}[1]{\widehat{#1}}

\nc{\ob}[1]{#1}

\def\beq{\begin {equation}}
\def\eeq{\end {equation}}

\def\be{\begin{equation}}
\def\ee{\end{equation}}

\nc{\eq}[1]{(\ref{eq:#1})} 
\nc{\eqs}[2]{\eq{#1} and \eq{#2}}

\nc{\eqn}[1]{Eq.~(\ref{eqn:#1})}
\nc{\eqns}[2]{Eqs.~(\ref{eqn:#1}) and (\ref{eqn:#2})}

\newcommand{\secref}[1]{Section~\ref{sec:#1}}

\newcommand{\lemref}[1]{Lemma~\ref{lem:#1}}
\newcommand{\thmref}[1]{Theorem~\ref{thm:#1}}
\newcommand{\thmrefs}[2]{Theorems~\ref{thm:#1} and \ref{thm:#2}}
\newcommand{\propref}[1]{Proposition~\ref{prop:#1}}

\newcommand{\corref}[1]{Corollary~\ref{cor:#1}}

\nc{\region}{\cS\cW}

\newtoggle{long}
\nc{\longonly}[1]{\iftoggle{long}{#1}{}}
\nc{\shortonly}[1]{\iftoggle{long}{}{#1}}
\nc{\longshort}[2]{\iftoggle{long}{#1}{#2}}

\usepackage{boxedminipage}

\newenvironment{mybox}
{\center \noindent\begin{boxedminipage}{1.0\linewidth}}
{\end{boxedminipage}\noindent}

\usepackage{fullpage}
\toggletrue{long}

\begin{document}

\title{Product-state Approximations to Quantum Ground States}
\shortonly{\subtitle{Extended Abstract\titlenote{A full version of this paper is available on the arXiv.}}}

\longshort{
\author{Fernando G.S.L. Brand\~ao\footnote{Department of Computer Science, University College London, email: {\tt fgslbrandao@gmail.com}}
\and Aram W. Harrow \footnote{
Center for Theoretical Physics, Massachusetts Institute of Technology,  email: {\tt aram@mit.edu}
}}
\date{\today\vspace{-1em}}
}
{\numberofauthors{2} 
\author{
\alignauthor
Fernando~G.S.L.~Brand\~ao\\
	\affaddr{Computer Science, University College London}\\
\email{fgslbrandao@gmail.com}
\alignauthor
Aram~W.~Harrow\\
\affaddr{Physics, MIT }\\
\email{aram@mit.edu}
}}

\date{\today}
\maketitle

\begin{abstract}
The local Hamiltonian problem consists of estimating the ground-state energy (given by the minimum eigenvalue) of a local quantum Hamiltonian. It can be considered as a quantum generalization of constraint satisfaction problems (CSPs) and has a central role both in quantum many-body physics and quantum complexity theory. A key feature that distinguishes quantum Hamiltonians from classical CSPs is that the solutions may involve complicated entangled states. In this paper, we demonstrate several large classes of Hamiltonians for which product (i.e. unentangled) states can approximate the ground state energy to within a small extensive error.


First, we show the {\em existence} of a good product-state approximation for the ground-state energy of 2-local Hamiltonians with one or more of the following properties: (1) high degree, (2) small expansion, or (3) a ground state with sublinear entanglement with respect to some partition into small pieces. The approximation based on degree is a surprising difference between quantum Hamiltonians and classical CSPs, since in the classical setting, higher degree is usually associated with {\em harder} CSPs. The approximation based on low entanglement, in turn, was previously known only in the regime where the entanglement was close to zero. 

Estimating the energy is $\NP$-hard by the PCP theorem, but whether it is $\QMA$-hard is an important open question in quantum complexity theory. A positive solution would represent a quantum analogue of the PCP theorem. Since the existence of a low-energy product state can be checked in $\NP$, the result implies that any Hamiltonian used for a quantum PCP theorem should have: (1) constant degree, (2) constant expansion, (3) a ``volume law'' for entanglement with respect to any partition into small parts. The result also gives a no-go to a quantum version of the PCP theorem \textit{in conjunction} with parallel repetition for quantum CSPs.

Second, we show that in several cases, good product-state approximations not only exist, but can be found in polynomial time: (1) 2-local Hamiltonians on any planar graph, solving an open problem of Bansal, Bravyi, and Terhal, (2) dense $k$-local Hamiltonians for any constant $k$, solving an open problem of Gharibian and Kempe, and (3) 2-local Hamiltonians on graphs with low threshold rank, via a quantum generalization of a recent result of Barak, Raghavendra and Steurer.

Our work introduces two new tools which may be of independent interest.  First, we prove a new quantum version of the de Finetti theorem which does not require the usual assumption of symmetry.  Second, we describe a way to analyze the application of the Lasserre/Parrilo SDP hierarchy to local quantum Hamiltonians.

\end{abstract}


\parskip .75ex



\section{Background}

A quantum $k$-local Hamiltonian on $n$ qudits is a $d^n \times d^n$ Hermitian matrix $H$ of the form
\begin{equation}
H = \frac{1}{l} \sum_{i=1}^l H_{i},
\end{equation}
where each term $H_i$ acts non-trivially on at most $k$ qudits\footnote{See \secref{notation} for a precise definition.} and $\Vert H_i \Vert \leq 1$. Local Hamiltonians are ubiquitous in physics, where interactions are almost always few-body. Of particular interest -- e.g. in quantum many-body physics, quantum chemistry, and condensed matter physics -- is to understand the low-energy properties of local Hamiltonians. These can be either the low-lying spectrum of the model, or properties -- such as correlation functions -- of the ground state of the Hamiltonian, defined as the eigenstate associated to the minimum eigenvalue. A benchmark problem, termed the local Hamiltonian problem, is to approximate the ground-state energy (i.e. the minimum eigenvalue) of the model. What is the computational complexity of this task?

In a seminal result Kitaev proved that the local Hamiltonian problem is complete for the quantum complexity class $\QMA$, in the regime of estimating the energy to within polynomial accuracy in the number of particles of the model \cite{KSV02}. The class $\QMA$ is the quantum analogue of $\NP$ in which the proof consists of a quantum state and the verifier has a quantum computer to check its validity \cite{Wat08}. Since then there have been several developments showing that simpler classes of local Hamiltonians are still hard \cite{KKR06, OT05, SV09}, culminating in the $\QMA$-completeness of the problem for Hamiltonians on a line \cite{AGIK09} even when all the local terms are the same and the only input is the length of the line expressed in unary~\cite{GI09}.
In fact a new area has emerged around the question of understanding the computational complexity of local Hamiltonians, at the crossover of quantum complexity theory and condensed matter physics: quantum Hamiltonian complexity \cite{Osb11}. 

It is an important observation that the $\QMA$-hardness of the local Hamiltonian problem implies, assuming $\QMA \neq \NP$, that one cannot find an efficient classical description of the ground state of local models in general. An efficient classical description is a representation of the quantum state in terms of polynomially many bits that would allow for the computation of local observables (e.g. the energy) in polynomial time. The $\QMA$-hardness results therefore say that ground states of certain local Hamiltonians can be highly non-classical objects. In particular they can be highly entangled states. Indeed the local Hamiltonian problem can be seen as a quantum analogue of constraint satisfaction problems (CSPs), with the key distinguishing feature of quantum Hamiltonians being that the solutions may involve complicated entangled states. 

A popular technique in physics to cope with this further level of complexity in the quantum case is to consider a product (i.e. unentangled) state as an \textit{ansatz} to the true state. In quantum many-body physics this approximation is commonly refereed to as \textit{mean-field theory}, while in the context of quantum chemistry it is also called \textit{Hartree-Fock method}. Conceptually, disregarding the entanglement in the ground state maps the quantum problem to a classical constraint satisfaction problem. 

In this paper we give several new results concerning the usefulness of this product-state approximation for a large class of Hamiltonians. One motivation is to obtain a deeper understanding of when entanglement is relevant and to develop new polynomial-time algorithms for estimating the ground energy of interesting classes of Hamiltonians. A second motivation is related to a central open problem in quantum Hamiltonian complexity: The existence of a quantum version of the PCP theorem.

\subsection{The quantum PCP conjecture}

Let ${\cal C}$ be a CSP and let $\text{unsat}({\cal C})$ be the fraction of the total number of local constraints that are not satisfied in the best possible assignment for ${\cal C}$. Then one formulation of the Cook-Levin theorem is that it is $\NP$-hard to compute $\text{unsat}({\cal C})$ for an arbitrary CSP ${\cal C}$. However this leaves open the possibility that one might be able to approximate $\text{unsat}({\cal C})$ to within small error in reasonable time. 

\longshort{The objective of the theory of approximation algorithms is to find out in which cases one can obtain in polynomial time approximate solutions to problems that are hard to solve exactly. This is by now a well-developed theory~\cite{Vaz03}.}{This question of efficient approximability has been well-studied for CSPs~\cite{Vaz03}.}
 The quantum case, in contrast, is largely unexplored, with a few notable exceptions \cite{BBT09, GK11, Ara11, Has12}. It is also interesting to \longshort{understand what are the limitations for obtaining approximation algorithms; this is the goal of the theory of hardness of approximation, where one is interested in proving the computational hardness of obtaining even approximate solutions}{prove hardness-of-approximation results}. A landmark result in this theory, on which almost all other hardness-of-approximation results are based, is the PCP (Probabilistically Checkable Proof) Theorem \cite{ALMSS98, AS98, Din07}. One possible formulation in terms of 2-CSPs\longonly{\footnote{a $k$-CSP is a CSP in which each constraint only involves $k$ variables.}} is the following:

\begin{thm} [PCP Theorem \cite{ALMSS98, AS98, Din07}]
There is a constant $\varepsilon_0 > 0$ such that it is $\NP$-hard to determine whether for a given 2-CSP ${\cal C}$, $\text{unsat}({\cal C}) = 0$ or $\text{unsat}({\cal C})  \geq \varepsilon_0$.
\end{thm}

\longonly{As well as developing a quantum theory of approximation algorithms, a theory of hardness of approximation for $\QMA$ would also be an interesting development.}  Inspired by the PCP theorem we might be tempted to speculate a quantum hardness-of-approximation result for the local Hamiltonian problem, which would correspond to a quantum analogue of the PCP theorem. In more detail, consider a local Hamiltonian $H = \E_i H_{i}$ where $\E_i$ refers to taking the expectation over $i$ with respect to some distribution, and ``local'' means that no $H_i$ acts on more than a constant number of qubits.  The ground-state energy is defined by the smallest eigenvalue of $H$:
\be e_0(H) := \min \{ \bra\psi H \ket \psi : \braket\psi\psi = 1\}.\ee
The ground-state energy is the quantum analogue of the unsat value of CSPs. A central open problem in quantum Hamiltonian complexity is the validity of the following  \cite{Osb11, Ara11, Has12, AALV09, Aar06, BDLT08, PH11, AAV13}:

\begin{con} [Quantum PCP Conjecture] \label{qpcpconjecture}
There is a constant $\varepsilon_0 > 0$ such that it is $\QMA$-hard to decide whether for a 2-local Hamiltonian $H$, $e_{0}(H) \leq 0$ or $e_{0}(H) \geq \varepsilon_0$.
\end{con}
We remark that $H$ is not necessarily positive semidefinite, so the case of $e_0(H)\leq 0$ does not necessarily correspond to perfect completeness.

One could also consider the conjecture for $k$-local Hamiltonians. However Conjecture \ref{qpcpconjecture} is known to be equivalent to this more general version, even when restricted to 2-local Hamiltonians with sites formed by qubits \cite{BDLT08}.

There are many reasons why the quantum PCP conjecture is interesting (although a caveat to the following is that the quantum PCP conjecture, like the classical PCP theorem, is known not to hold in any constant number of spatial dimensions~\cite{BBT09}). 
\benum
\item Like the classical PCP theorem, a quantum PCP theorem would yield a hardness-of-approximation result for a $\QMA$-hard problem at the error scale that is most natural. Moreover one could expect that, similarly to the PCP theorem, a quantum PCP theorem would lead to hardness-of-approximation results for other QMA-hard problems (e.g. see the survey~\cite{Bookatz-QMA}).
 \item One of the central goals of theoretical, and especially condensed-matter, physics, is to explain how the properties of a quantum system, starting with its ground-state energy, emerge from its underlying Hamiltonian.
A common method for doing so is to use insight and inspired guesswork to come up with an {\em ansatz}, which is an approximate ground state with a succinct classical description.  If the quantum PCP conjecture were true, this would mean that there would exist families of Hamiltonians for which succinct, efficiently-verifiable, ansatzes do not exist, even if we accept constant extensive error.  A related question is the NLTS (``no low-energy trivial states'') conjecture from \cite{FreedmanH13}, which can be equivalently stated as ruling out the specific family of ansatzes obtained by applying a constant-depth circuit to a fixed product state.
\item A quantum PCP conjecture would imply the $\QMA$-hardness of approximating the energy of thermal states of local Hamiltonians at \textit{fixed} (but sufficiently small) temperature.  However on physical grounds one might expect that genuine quantum effects should be destroyed by the thermal fluctuations.  A resolution of the quantum PCP conjecture would give new insight into this question. 
\eenum

See \cite{AAV13} for a recent review on the significance and status of the quantum PCP conjecture.

A way to \textit{refute }the quantum PCP conjecture is to show that the estimation of the ground-state energy can be realized in $\NP$ \footnote{Note it it extremely unlikely that this estimation can be done in $\P$ or even in $\BQP$, since we know it is $\NP$-hard by the PCP theorem.}. In this work we follow this approach and give new approximation guarantees for the local Hamiltonian problem in the regime of extensive error.   

\subsection{Notation} 
\label{sec:notation}
Let ${\cal D}({\cal H})$ be the set of quantum states on ${\cal H}$, i.e. positive semidefinite matrices of unit trace acting on the vector space ${\cal H}$.  A state on systems $ABC$ is sometimes written $\rho^{ABC}$ and we write $\rho^A$ to denote the marginal state on the $A$ system, i.e. $\tr_{BC}\left( \rho^{ABC} \right)$.  For a pure state $\ket\psi$, we use $\psi$ to denote the density matrix $\proj\psi$. 

We will measure entropy in ``nats'' so that the entropy of a density matrix $\rho$ is defined as $S(\rho) := -\tr[\rho\ln(\rho)]$.  Entropic quantities are often written in terms of the subsystems; i.e. if $\rho$ is a state on $ABC$, then $S(ABC)_\rho := S(\rho^{ABC})$, $S(A)_\rho := S(\rho^A)$, etc.  When the state is understood from context, we may suppress the subscript.  Define the mutual information $I(A:B)= S(A) + S(B)-S(AB)$ and the conditional mutual information $I(A:B|C) = S(AC)+S(BC)-S(ABC)-S(C)$.   Additional information-theoretic quantities are defined in \secref{mutual-info}.

Let $\Vert X \Vert_1 := \tr \sqrt{X^\dag X}$ be the trace norm of $X$, and let $\|X\|$ denote the operator norm, or largest singular value, of $X$.


Given a graph $G = (V, E)$ and a subset of vertices $S$, we define the expansion of $S$ as 
\begin{equation}
\Phi_G(S) := \Pr_{(u, v) \in E}\left[ v \notin S | u \in S  \right].
\end{equation}
The expansion of  the graph $G$ is defined as
\be
\Phi_G := \min_{S : |S| \leq |V|/2} \Phi_G(S).
\ee

Let $H_d$ denote the set of Hermitian $d\times d$ matrices.
Define the set of $k$-local Hamiltonians to be
\begin{align}
W_k =\shortonly{\nn&} W_k^{(d,n)} = 
\Span \{ A_1 \ot \cdots \ot A_n : A_1,\ldots, A_n\in H_d,
\shortonly{\\&} \text{at most $k$ of the $A_i$ are $\neq I_d$}\}.
\label{eq:k-local-def}
\end{align}
In the language of quantum error-correction, the $W_k$ are the
operators of weight $\leq k$.  For those familiar with Fourier
analysis of Boolean functions, the diagonal matrices in $W_k$ are
precisely the functions of $\{0,1\}^n$ with all nonzero Fourier
coefficients of weight $\leq k$.  

Say that a state $\sigma$ is globally separable if it can be written as $\sum_z p_z \sigma_z$ for some $\sigma_z$ of the form $\sigma_z = \sigma_z^{Q_1} \ot \ldots \ot \sigma_z^{Q_n}$ and a probability distribution $p_z$. A state is separable with respect to a partition $V_1:\ldots:V_m$ if it can be written in the form $\sum_z p_z \sigma_z^{V_1} \ot \cdots \ot \sigma_z^{V_m}$.

\section{Results}

Our main results concern Hamiltonians for which product states give good approximations of the ground-state energy.  In \secref{NP-approx}, we demonstrate a large class of Hamiltonians with this property, which puts the corresponding ground-state estimation problem in $\NP$.  

In \secref{P-approx}, we give three cases in which such a product-state approximation can be found efficiently: planar Hamiltonians in \secref{planar}, dense Hamiltonians in \secref{dense} and Hamiltonians on graphs with low threshold rank (defined below) in \secref{threshold}. 

Each of these results, except for the case of low threshold rank (which follows the analysis of the Lasserre hierarchy in \cite{BRS11}) uses an information-theoretic approach based on the chain rule of conditional mutual information, and is inspired by a method due to Raghavandra and Tan~\cite{RT12} originally developed to prove fast convergence of the Lasserre/Parrilo/sum-of-squares (SOS) hierarchy for certain CSPs. Similar techniques are also used to prove several new quantum de Finetti theorems in our companion paper~\cite{BH-local}. 

Indeed, de Finetti theorems are another sort of product-state approximation and our proof approach can be seen as a de Finetti theorem. We explore this in \secref{deF} by stating a new de Finetti theorem which is the first version of the de Finetti theorem to \textit{not} require any symmetry assumptions.

 The low-threshold-rank result in \secref{threshold} is based on a variant of the SOS hierarchy that has been widely used in quantum chemistry~\cite{Erdahl78, NY96, YN97, Mazz04,BarthelH12,BaumgratzP12}, but until now only as a heuristic algorithm.

\subsection{Approximation Guarantees in $\NP$}\label{sec:NP-approx}

Our first and main result shows that one can obtain a good approximation to the ground-state energy (or indeed the energy of any state) by a product state for a large class of 2-local Hamiltonians. To every 2-local Hamiltonian we can associate a constraint graph $G = (V, E)$ with $V$ the set of vertices associated to the sites of the Hamiltonian, and $E$ the set of edges associated to the interaction terms of the Hamiltonian.  The basic result is that if $G$ is a $D$-regular graph then product states can approximate the ground energy with an error that is a power of $1/D$, yielding a good approximation for sufficiently large constant degree $D$. 

The rest of this section is organized as follows.  The basic version of our result described above is given in \secref{basic}.  Then we proceed to two different generalizations.  In \secref{clustered}, we consider $D$-regular graphs in which groups of vertices have been clustered together.  In this case, our approximation is further improved if the clusters are less than maximally expanding, or are less than maximally entangled. Then in \secref{weighted}, we give a version of the theorem for variable-degree weighted graphs in which $1/D$ is replaced by the collision probability (probability that a two-step random walk starting from a random vertex will return to its starting point). Then in \secref{discussion} we discuss the significance of the results for the quantum PCP conjecture and connect them to previous work. An overview of the proof is given in \secref{sketch}.  Finally in \secref{variants} we discuss some additional variants of main theorem.

\subsubsection{Regular graphs} \label{sec:basic}
\nc\ThmBasic[1]
{  Let $G=(V,E)$ be a $D$-regular graph with $n=|V|$.  Let $\rho^{Q_1 ... Q_n}$ be an $n$-qudit state.  Then there exists a globally separable state $\sigma$ such that
\be \E_{(i,j)\in E} \|\rho^{Q_i Q_j} - \sigma^{Q_i Q_j}\|_1 \leq 
12\L(\frac{d^2\ln(d)}{D}\R)^{1/3}
\label{eq:basic-bound-#1}.\ee
}
\begin{thm}[Basic product-state approximation]\label{thm:basic}
\ThmBasic{first}
\end{thm}

\thmref{basic} applies to any state $\rho$. If $\rho$ is the ground state of a Hamiltonian then we obtain:
\begin{cor}\label{cor:basic}
Let $G=(V,E)$ be a D-regular graph, and let $H$ be a Hamiltonian on $n=|V|$ $d$-dimensional particles given by 
$$H = \E_{(i,j)\in E} H_{i,j} = \frac{2}{nD} \sum_{(i,j)\in E} H_{i,j},$$
where $H_{i,j}$ acts only on qudits $i,j$ and satisfies $\|H_{i,j}\|\leq 1$.
Then there exists a state $\ket\varphi$ such that $\ket\varphi = \ket\varphi^{Q_1} \ot \ldots \ot
\ket\varphi^{Q_n}$ and
\be
 \tr(H\varphi) 
\leq e_0(H) + 12\L(\frac{d^2\ln(d)}{D}\R)^{1/3} \ee
\end{cor}

\begin{cor}\label{cor:NP}
Let $G,H$ be as in \corref{basic}.  Consider the promise problem of estimating whether $e_0(H)\leq \alpha$ or $e_0(H)\geq \beta$.  If $\beta-\alpha \geq 12\L(\frac{d^2\ln(d)}{D}\R)^{1/3}+ \delta$ then there exists witnesses of length $O(n d \log(n/\delta))$ for the $e_0(H)\leq \alpha$ case.  The verifier requires time $nD d^4\poly\log(1/\delta)$.
\end{cor}

\subsubsection{Clustered approximation for regular graphs}\label{sec:clustered}
The first refinement of \thmref{basic} we consider is when the graph admits a natural way to group the vertices into clusters. In case {\em either} the clusters are lightly expanding or are lightly entangled, we can improve on the bound of \thmref{basic}.  

\nc\ThmClustered[1]{
Let $G=(V,E)$ be a D-regular graph, with $n=|V|$, $m$ a positive integer dividing $n$ and $\{V_1,\ldots, V_{n/m}\}$ a partition of $V$ into sets each of size $m$.  Define $\bar\Phi_G := \E_{i\in [n/m]} \Phi_G(V_i)$ to be the average expansion of the blocks.

Let $\rho$ be a state on $(\bbC^d)^{\ot n}$ and define $\bar I := \E_{i\in[n/m]} I(V_i:V_{-i})_\rho$ to be the average mutual information of $V_i$ with its complement for state $\rho$.

Then there exists a state $\sigma$ such that
\benum \item
$\sigma$ is separable with respect to the partition $(V_1,\ldots,V_{n/m})$; i.e. $\sigma = \sum_z p_z \sigma_z$ for some states $\sigma_z$ of the form $\sigma_z = \sigma_z^{V_1} \ot \ldots \ot \sigma_z^{V_{n/m}}$; and
\item
\be \E_{(i,j)\in E} \|\rho^{Q_i Q_j} - \sigma^{Q_i Q_j}\|_1
\leq 9 \L(\frac{d^2\bar\Phi_G}{D} \frac{\bar I}{m}  \R)^{1/6} + \frac 1 m + \frac m n 
\label{eq:clustered-#1}\ee
\eenum
}

\begin{thm}[Clustered product-state approximation for regular graphs]\label{thm:clustered}
\ThmClustered{first}
\end{thm}

Along the lines of corollaries \ref{cor:basic} and \ref{cor:NP}, we can apply \thmref{clustered} to the ground state of a local Hamiltonian.
\begin{cor}\label{cor:clustered}
Let $G$ and $V_1,\ldots,V_{n/m}$ be as in \thmref{clustered}, and let $H$ be a Hamiltonian on $n=|V|$ $d$-dimensional particles given by 
$$H = \E_{(i,j)\in E} H_{i,j} = \frac{2}{nD} \sum_{(i,j)\in E} H_{i,j},$$
where $H_{i,j}$ acts only on qudits $i,j$ and satisfies $\|H_{i,j}\|\leq 1$.
Let $\psi_0 = \proj{\psi_0}$ be a ground state of $H$; i.e. $e_0(H) = \bra{\psi_0}H\ket{\psi_0}$.  
Then there exists
$\ket\varphi := \ket{\varphi_{1}}^{V_1} \ot  \ldots \ot \ket{\varphi_{n/m}}^{V_{n/m}}$ such that
\be
 \tr(H\varphi) 
\leq e_0(H) + 9
\L(\frac{d^2\bar\Phi_G}{D}\frac{\mathbb{E}_{i} S(V_i)_{\psi_0}}{m} \R)^{1/6}
+ \frac 1 m + \frac m n.\ee
There is  a classical string of length $O(n 2^m \log(1/\delta))$
describing $\ket\varphi$ to accuracy $\delta$.   
\end{cor}

This description can be verified in time polynomial in the witness length.  Thus, the resulting ground-state estimation problem is in $\NP$ when $m=O(\log n)$ (or more generally has a subexponential-size witness when $m=o(n)$).

\subsubsection{Non-regular graphs:}\label{sec:weighted}
 In Theorems \ref{thm:basic} and \ref{thm:clustered} our requirement that $G$ be a regular graph was somewhat artificial.  Indeed, we can also consider graphs with weighted edges.  More generally, define
\be H := \sum_{(i,j)\sim G} H_{i,j} := \sum_{i,j} G_{i,j} H_{i,j},
\label{eq:weighted-graph}\ee
 where each $\|H_{i,j}\|\leq 1$ and $G$ is a matrix satisfying $G_{i,j}\geq 0$ and $\sum_{i,j} G_{i,j} = 1$.
This is ``PCP-style'' weighting with the triangle inequality naturally yielding the bound $\|H\|\leq 1$, and with the quantum PCP conjecture concerning the hardness of an $\eps$-approximation to $e_0(H)$ for sufficiently small constant $\eps$.

We can interpret $G$ as a graph as follows.  Without loss of generality we can assume that $G=G^T$.
Define the vector $\pi$ by $\pi_j := \sum_i G_{i,j}$, and observe that $\pi$ is a probability distribution.  Define $A_{i,j} = G_{i,j} / \pi_j$.  Since $\sum_i A_{i,j}=1$ for any $j$, $A$ is a stochastic matrix.  Further observe that $A\pi = \pi$.  Thus we can interpret $A$ as a random walk with stationary distribution $\pi$, and we can interpret $G$ as the distribution on edges obtained by sampling a vertex according to $\pi$ and then taking one step of the walk according to $A$.  Because $G=G^T$, the random walk is reversible \footnote{Despite this rather specific interpretation, our construction is of course fully general.  Indeed, suppose that $\tilde H = \sum_{i,j} \tilde H_{i,j}$.  Let $\Gamma = \sum_{i,j} \|\tilde H_{i,j}\|$.  Then we can define $G_{i,j} = \|\tilde H_{i,j}\|/\Gamma$, $H_{i,j} = \frac{\tilde H_{i,j}}{\|\tilde H_{i,j}\|}$ and we have that $G$ is of the appropriate form and $H := \tilde  H/\Gamma = \sum_{i,j} G_{i,j} H_{i,j}$.}.

In this case we can prove an analogue of \thmrefs{basic}{clustered} for weighted graphs, but without the ability to define clusters.  Thus, our guarantee will not be in terms of expansion or entanglement, but only a quantity analogous to average degree.  

\def\ThmWeighted#1{
Let $G$ be a symmetric matrix with nonnegative entries summing to one and let $A_{i,j} := G_{i,j} / \sum_{i'} G_{i',j}$.  Then for any $\rho \in \cD((\bbC^d)^{\ot n})$ there exists a globally separable state $\sigma$ such that
\be
\E_{(i,j)\sim G} \|\rho^{Q_i Q_j} - \sigma^{Q_iQ_j}\|_1 \leq
14 \L(d^4\ln(d)\tr[A^2]\|\pi\|_2^2\R)^{1/8} + \|\pi\|_2^2
\label{eq:weighted-#1}\ee
If further we have a Hamiltonian $H$ given by Eq. \eq{weighted-graph}, then there exists a state
 $\ket{\varphi} = \ket{\varphi_1} \ot \cdots \ot \ket{\varphi_n}$ such that
\be  \tr (H\varphi) \leq e_0(H) + 
14 \L(d^4\ln(d)\tr[A^2]\|\pi\|_2^2\R)^{1/8} + \|\pi\|_2^2
\label{eq:H-weighted-#1}
\ee
}
\begin{thm}\label{thm:weighted}
\ThmWeighted{first}
\end{thm}
The proof is in \secref{proof-weighted}.

Degree here is replaced by $\|\pi\|_2^2 \tr(A^2)$ and indeed it can be seen as the harmonic mean of the degree.  Ideally the bound would be unchanged by trivial changes such as adding degree-0 vertices.   One possible such strengthening of \thmref{weighted} would be to replace the average degree with the harmonic mean weighted by the measure $\pi$; i.e. $1/\sum_{i,j} \pi_i A_{i,j}^2$.
\begin{con}\label{con:weighted}
The $\tr[A^2]\|\pi\|_2^2$ term in \thmref{weighted} can be replaced by some power of $\sum_{i,j} \pi_i A_{i,j}^2$.
\end{con}

\subsubsection{Discussion of the Results and Implications for the Quantum PCP Conjecture} \label{sec:discussion}

Let us discuss the significance of Theorems \ref{thm:weighted}, \ref{thm:clustered}, and \ref{thm:basic}. \thmref{clustered} gives a non-trivial approximation in terms of three parameters: the average expansion $\bar\Phi_G$ of the regions $V_i$, the degree $D$ of the graph, and the average entanglement $\bar I$ of the regions $V_i$ with their complement.  The results also get weaker as the local dimension $d$ increases.  Below we discuss the relevance of each of them.

\noindent \textbf{Average Expansion:} This is the least surprising term, since one can always delete the edges between different blocks and incur an error that is upper bounded by the average expansion.

\vspace{0.2 cm}

\noindent \textbf{Degree of the Graph:} We believe this is a more surprising approximation parameter. Indeed an easy corollary of the PCP theorem \cite{ALMSS98, AS98, Din07} is the following (see section \ref{sec:proofCSPs} for a proof):

\def\propCSPs{
For any constants $c, \alpha, \beta > 0$, it is $\NP$-hard to determine whether a 2-CSP ${\cal C}$, with alphabet $\Sigma$ on a constraint graph of degree $D$, is such that $\text{unsat}({\cal C}) = 0$ or $\text{unsat}({\cal C}) \geq c |\Sigma|^{\alpha}/D^{\beta}$.
}

\begin{prop} \label{prop:CSPs}
\propCSPs
\end{prop}

Assuming $\QMA \nsubseteq \NP$, \thmref{clustered} shows the analogous result for the quantum case -- namely that it is $\QMA$-hard to achieve a $c d^{\alpha}/D^{\beta}$ approximation for the ground-state energy of 2-local Hamiltonians -- is \textit{not} true. This suggest a route for a disproof of the quantum PCP conjecture by trying to a find a procedure for increasing the degree of the constraint graph without decreasing the ground-state energy of the model (by possibly increasing the dimension of the particles as well). 

Classically this can be achieved by parallel repetition of the CSP~\cite{Raz98} or by gap amplification~\cite{Din07}, but it is unclear whether there is an analogous construction in the quantum case. We note that such degree-increasing maps have been pursued in the quantum case with an eye toward proving the quantum PCP conjecture~\cite{AALV09} whereas our results suggest that they may be useful in {\em disproving} it. Indeed a consequence of the result is that a quantization of the the gap amplification part of Dinur's proof of the PCP theorem \cite{Din07}  to the quantum setting would disprove the quantum PCP conjecture. Therefore we conclude that we cannot quantize Dinur's approach without introducing substantial new ideas.


Parallel repetition gives a mapping ${\cal P}_t$ from 2-CSP to 2-CSP, for every integer $t$, such that $\mathrm{deg}({\cal P}_t({\cal C})) = \mathrm{deg}({\cal C})^t$, $|\Sigma_{{\cal P}_t({\cal C})}| = |\Sigma_{{\cal C}}|^t$ \footnote{$\Sigma_{{\cal P}_t({\cal C})}$ and $\Sigma_{{\cal C}}$ are the alphabets of ${\cal P}_t({\cal C})$ and ${\cal C}$, respectively.}, $\text{unsat}({\cal P}_t({\cal C})) \geq \text{unsat}({\cal C})$, and $\text{unsat}({\cal P}_t({\cal C})) = 0$ if $\text{unsat}({\cal C}) = 0$  (see section \ref{sec:proofCSPs} for more details). The next corollary shows that a similar notion of parallel repetition for quantum constraint satisfaction problems would disprove the quantum PCP conjecture.

\nc\CorAmplification{
Suppose that for every $t \geq 1$ there is a mapping ${\cal P}_t$ from 2-local Hamiltonians on $d$-dimensional particles to 2-local Hamiltonians on $n_t$ $d_t$-dimensional particles such that:
\begin{enumerate}[label=(\roman*)]
\item $\cP_t$ can be computed deterministically in polynomial time.
\item $\cP_t(H) = \E_{(i,j)\sim G'} H'_{i,j}$ with $G'$ a probability distribution and $\|H'_{i,j}\|\leq 1$.
\item $\mathrm{deg}({\cal P}_t(H)) \geq \mathrm{deg}(H)^t$,
\item $n_t \leq n^{O(t)}$,
\item $d_t = d^t$,
\item $e_0({\cal P}_t(H)) \geq e_0(H)$, if $e_{0}(H) > 0$,
\item  $e_0({\cal P}_t(H)) \leq e_0(H)$, if $e_{0}(H) \leq 0$.
\end{enumerate}

Then the quantum PCP conjecture is false. 

}

\begin{cor}\label{cor:amplification}
\CorAmplification
\end{cor}
The proof is in \secref{proofCSPs}.

 \textit{Entanglement Monogamy and Mean Field:} The approximation in terms of the degree can be interpreted as an instance of monogamy of entanglement \cite{Ter03}, a feature of quantum correlations that says a system cannot be highly entangled with a large number of other systems. If the degree is high, it means that every particle is connected to many others by an interaction term. But since it cannot be highly entangled with many of them, we might expect that a product state gives a reasonable approximation to the energy of most such interactions. Indeed based on this idea it is folklore in condensed matter physics that mean field (where one uses a product state as an ansatz to the ground state) becomes exact once the dimension of the interaction lattice increases. For rotationally symmetric models this can be rigorously established \cite{CKL13} using the standard quantum de Finetti Theorem \cite{CKMR07}. The approximation we obtain here in terms of the degree of the graph gives a rigorous justification in the general case.

 \textit{Highly Expanding Graphs:} Finally the result also gives an approximation for \textit{highly expanding} graphs. Using the bound $\Phi_G \leq 1/2 - \Omega(\mathrm{deg}^{-1/2})$ (see e.g. Ref. \cite{HLW06}), \thmref{clustered} gives the existence of states $\ket{\varphi_1}, \ldots, \ket{\varphi_n}$ such that
\begin{equation}
\bra{\varphi_1, \ldots, \varphi_{n}} H \ket{\varphi_1, \ldots, \varphi_{n}} \leq e_{0}(H) + \Omega \left[ \left( \frac{1}{2} - \Phi_G \right)^{1/3} \right].
\end{equation}
Therefore very good expanders, with $\Phi_G$ approaching the maximum value $1/2$, are not candidates for $\QMA$-hard instances.  See also \secref{threshold} for a similar setting in which not only do product states provide a good approximation, but a good choice of product state can be efficiently found.

A recent and independent result of Aharonov and Eldar~\cite{AE13} shows that the ground-state energy of commuting Hamiltonians on hypergraphs with almost perfect small-set expansion can be approximated in $\NP$.  This is weaker than our result in the sense of its restriction to commuting Hamiltonians, and its requirement of high small-set expansion (which is similar to, but strictly more demanding than our requirement of high degree), but it is stronger than our result in its applicability to $k$-local Hamiltonians with $k>2$.  To apply our proof techniques to $k$-local Hamiltonians for $k>2$ and obtain error $\eps$, we would need there to be $n^{k-1}/\eps^{O(1)}$ terms, so only for $k=2$ could this be interpreted as ``large constant degree.''  Another possible way to obtain the results of \cite{AE13} from ours would be to reduce a $k$-local Hamiltonian to a 2-local Hamiltonian using the perturbation-theory gadgets of \cite{BDLT08}.  Unfortunately the resulting 2-local Hamiltonian has a constant fraction of vertices with only constant degree, which means our results do not yield nontrivial bounds here.  Indeed the gadgets of \cite{BDLT08} create a ground state with significant entanglement, so any classical approximation of it would need to go beyond using product states.

\vspace{0.2 cm}

\noindent \textbf{Average Entanglement:} Given a bipartite pure state $\ket{\psi}^{AB}$ a canonical way to quantify the entanglement of $A$ and $B$ is by the von Neumann entropy of the reduced density matrix $\psi^A$ (or $\psi^B$), given by $S(A)^{\psi} := - \tr(\psi^A \ln \psi^A)$ \cite{HHHH09}. \thmref{clustered} shows that we can obtain an ansatz of size $n2^{O(m)}$ that approximates well the ground energy whenever the ground state satisfies a \textit{subvolume} law for entanglement, i.e. whenever  we can find a partition of the sites into regions $V_i$ of size $m$ such that in the ground state their entanglement (with the complementary region) is sublinear, on average, in the number of particles of the regions.

Although the dependence of our result on entanglement is intuitive, it is not trivial.  One can trivially say (from Pinsker's inequality) that if $\E_i S(V_i)_\psi \leq \eps$ then there will exist a product state with energy within $O(\sqrt\eps)$ of that of $\psi$.  However, we are very far from this regime, and obtain effective approximations whenever the entanglement is merely not within a constant factor of its maximum value.

\longonly{There is a general intuition that the amount of entanglement in the ground state of a Hamiltonian should be connected to its computational complexity (see e.g. \cite{ECP10}). For one-dimensional models this can be made precise: An area law for entanglement in the ground state, meaning that the entanglement \footnote{In this case one must quantify entanglement not by the von Neumann entropy, but by the smooth max-entropy \cite{SWVC08}.} of an arbitrary connected region is a constant, implies there is an efficient classical description for estimating the ground-state energy to within inverse polynomial accuracy \cite{Vid03, VC06, SWVC08} (in terms of a matrix-product state of polynomial bond dimension \cite{FNW92}). 

Here we find a similar situation, but with notable differences. On one hand the result is weaker in that it only implies an efficient classical description of the ground state for estimating the ground-state energy; on the other hand the result is much more general, being true for 2-local Hamiltonians on general interaction graphs and only requiring a subvolume law instead of an area law. This new relation of entanglement and computational hardness gives a motivation for studying the entanglement scaling in ground states of Hamiltonians on more general graphs than finite dimensional ones. 

Another interesting choice for the state $\rho$ in \thmref{clustered} is the thermal state of the Hamiltonian at temperature $T$: $\rho_{T} := \exp(- H / T) / \tr(\exp(- H / T))$.  In this case, it is more convenient to normalize the Hamiltonian to have an extensive amount of energy; i.e. $H = \sum_{i \in V} \E_{j \sim \Gamma(i)} H_{ij}$, where $\Gamma(i)$ is the set of neighbors of $i$ and again each $\|H_{i,j}\|\leq 1$.  Thus $\|H\|\leq n$. Now
using the bound\footnote{This bound can be derived from the fact that the Gibbs state $\rho_T$ minimizes the free energy $F(\rho):=\tr(H\rho) - T S (\rho)$.} $\tr(H \rho_{T}) \leq e_{0}(H) + \varepsilon$ for $T = O( \varepsilon/ \ln(d))$ we find there is a state $\ket\varphi = \ket{\varphi_1} \ot \ldots\ot \ket{\varphi_{n/m}}$ satisfying
\be
\frac{1}{n}\tr (H\varphi) \leq 
\frac 1 n e_{0}(H)  + \varepsilon + 9 \L(\frac{d^2\bar\Phi_G}{D}\frac{\E_i I(V_i:V_{-i})_{\rho_T}}{m}\R)^{1/6}.
\ee

Therefore Hamiltonians for which the thermal state satisfies a subvolume law for the mutual information for constant, but arbitrary small, temperatures are also not candidates for $\QMA$-hard instances. We find this observation interesting since it is much easier to prove area laws for thermal states than for ground states \cite{WVHC08}. We note however that the known area law for thermal states from Ref. \cite{WVHC08} is not strong  enough to give a non-trivial approximation for arbitrary 2-local Hamiltonians.

The approximation in terms of the average entanglement means that in order for the quantum PCP conjecture to be true we must find a family of Hamiltonians on qubits whose ground states satisfy a strict volume law for the entanglement of every region of size $n^{o(1)}$ (on average over the choice of the region for every partition of the sites into $n^{o(1)}$-sized regions).  Such examples can be constructed: For example, take a 3-regular expander and place three qubits at each vertex.  Using 2-local interactions we can ensure that the ground state is the tensor product of an EPR pair along each edge.  The expansion condition then ensures that this state satisfies a volume law.  On the other hand such states are easy to give witnesses for using the techniques of \cite{BV05,Has12}.  To the best of our knowledge there is no known example of a Hamiltonian that both satisfies a strict volume law and cannot be described by a witness from \cite{Has12}.}

\noindent \textbf{Local dimension:}  Our error bound contains a term that grows as $d^{\Omega(1)}$ due mostly to the loss incurred by state tomography (cf. \lemref{info-complete}).  Since the information-theoretic parts of the argument involve terms that scale like $\log(d)$, it is worth asking whether the $d$-dependence of our result can be improved.  However, a simple example shows that a bound of the form $\log^\alpha(d)/D^\beta$ for $\alpha,\beta>0$ is impossible.  The example is none other than the ``universal counter-example'' that is ubiquitous in quantum information.

Let $d=n$ and define $F_{i,j}$ to be the swap operator on systems $i$ and $j$.  If $H = \E_{i\neq j} F_{i,j}$, then its ground state is the antisymmetric state
$$\ket\psi = \frac{1}{\sqrt{n!}}\sum_{\pi\in S_n} \sgn(\pi) \ket{\pi(1)} \ot \cdots \ot \ket{\pi(n)}.$$
While $\tr[\psi H] = -1$ it is easy to show that $\tr[\varphi H]\geq 0$ for all product states $\varphi$.  However here $D=d=n$.  Thus, for this parameter regime, product states cannot be arbitrarily accurate.

\subsubsection{Proof sketch}\label{sec:sketch}
 The proofs of Theorems \ref{thm:basic}, \ref{thm:clustered} and \ref{thm:weighted} are based on information-theoretic methods.  These methods were introduced by Raghavendra and Tan \cite{RT12} (see also \cite{BRS11}) and we adapt them to our purposes.  (The motivation of \cite{RT12,BRS11} was to analyze the Lasserre hierarchy for certain CSPs, and we explore a quantum version of the Lasserre hierarchy in \secref{threshold}.)

The key idea is that given a quantum state of $n$ subsystems $Q_1, \ldots, Q_n$, by measuring a small number $k$ of the subsystems chosen at random, one can show that {\em on average} the remaining pairs of subsystems will be close to a product state up to small error.  (Speaking about the average pair is crucial, as one can see by imagining $n/2$ Bell pairs.) This statement is a consequence of the \textit{chain rule} of the mutual information. As we show, the error scales as the quotient of the average mutual information of $V_i$ with its complementary region and $k$, which we choose as $k = \delta n$ for small $\delta > 0$. Then in a next step we show how an average pair of particles being close to product implies that an average pair of particles \textit{connected} by the constraint graph of the Hamiltonian are close to product (increasing the error in the process); it is here that the degree of the graph appears as an approximation factor.  So far, our strategy generally follows the iterative conditioning approach of \cite{BRS11}, which applied to classical random variables.  To bridge the gap to quantum states, we will need to use an informationally-complete measurement, which simply means a POVM from whose outcomes the original state of a system can be reconstructed.

For the product-state approximation,
we choose the tensor product of all single-particle reduced states of $\tilde\rho$, with the state $\tilde\rho$ being the postselected state obtained by measuring $t$ subsystems of the original state $\rho$ (which we typically take to be the ground state of the Hamiltonian). By the ideas outlined above, we can show that this state has energy no bigger than the R.H.S. of Eq. \eq{clustered-first}.

As a final ingredient in the clustered case (\thmref{clustered}), before applying the argument above we fuse some of the original particles of the model into blocks $V_1, \ldots, V_{n/m}$ which we treat as individual particles. This creates new difficulties (such as the fact that the degree of the effective graph on the blocks is not the same as the original degree) that are handled by applying the chain rule of mutual information in two stages: first to the blocks of sites and then to the sites in each block. 

\subsubsection{Variants of the main theorems}\label{sec:variants}
In this section we sketch a few easy-to-prove variants of our main theorems.

Another way to state our main theorems is that if we measure a random small subset of qudits and condition on the outcomes, then the remaining qudits will be approximately product across a random edge of the graph.  See Eq. \eq{q-trace-dist} for the clustered version of this claim.

This version is also meaningful for classical probability distributions, in which case we obtain a variant of Lemma 4.5 of \cite{RT12}.  The result, roughly, is that conditioning on a random $\delta n$ variables leaves {\em all} of the random variables approximately independent across a random edge.  Now there is no problem from measurement disturbance, although conditioning may have its own costs.

All of our results make use of state tomography because we wish to consider arbitrary 2-local interactions.  This is what leads to the $\poly(d)$ terms in the error.  However, suppose the individual Hamiltonian terms $H_{i,j}$ are of the form $\sum_{x,y} \lambda^{i,j}_{x,y} A_x \ot A_y$ where each $A_x \geq 0$, $\sum_x A_x = I$ and $|\lambda^{i,j}_{x,y}|\leq 1$.  Then we can replace the generic state tomography procedure with the measurement $\{A_x\}$.  This is essentially the approach used in \cite{BH-local} to prove de Finetti theorems for local measurements.  The result is to remove the $\poly(d)$ terms in our error bounds.

The above claims are straightforward to prove, although we will not do so in this paper.  We also state one conjecture that we cannot prove.  It is natural to conjecture a common generalization of \thmref{clustered} and \thmref{weighted} that yields a meaningful bound for clustered weighted graphs.  We see no barrier in principle to proving this, but the purely technical challenges are formidable and might obscure some more fundamental difficulty.

\subsection{Efficient algorithms for finding product-state approximations}\label{sec:P-approx}
When our task is only to find a product state minimizing the energy of a $k$-local Hamiltonian, then essentially our problem becomes a classical $k$-CSP, albeit with an alphabet given by the set of all unit vectors in $\bbC^d$.  However, we find that the problem is easy to approximate in many of the same cases that classical CSPs are easy to approximate, such as planar graphs, dense hypergraphs and graphs with low threshold rank.

\subsubsection{Polynomial-time Approximation Scheme for Instances on a Planar Graph}\label{sec:planar}

It is known that there are polynomial-time approximating schemes (PTAS) for a large class of combinatorial problems (including many important $\NP$-hard problems) on planar graphs \cite{Bak94}. Bansal, Bravyi, and Terhal generalized this result to obtain a PTAS for the local Hamiltonian problem on planar 2-local Hamiltonians with a bounded degree constraint graph \cite{BBT09}. They left as an open problem the case of unbounded degree. Combining their result with the techniques used in \secref{NP-approx} we show:

\def\Planar{
For every $\varepsilon > 0$ there is a polynomial-time randomized algorithm such that, given a 2-local Hamiltonian $H$ on a planar graph, computes a number $x$ such that $e_{0}(H) \leq x \leq e_0(H) + \varepsilon$.
}

\begin{thm} \label{thm:planar}
\Planar 
\end{thm}

The idea for obtaining Theorem \ref{thm:planar} is to use the same information-theoretic techniques behind \thmref{basic} to replace the ground state by a state which is product on all high-degree vertices and the original ground state on the low-degree part, only incurring a small error in the energy.  Then using the methods of \cite{Bak94,BBT09} we can also replace the state on the low-degree vertices by a tensor product of states each on a constant number of sites. At this point we have a constraint satisfaction problem over a planar graph and we can employ e.g. \cite{Bak94, BBT09} to approximate the optimal solution in polynomial time.

\subsubsection{Polynomial-time Approximation Scheme for Dense Instances}\label{sec:dense}

We say a $k$-local Hamiltonian $H = \frac{1}{m}\sum_{j=1}^m H_j$ is dense if $m = \Theta(n^k)$. Several works have established polynomial-time approximation schemes (PTAS) for dense CSPs \longshort{\cite{VKKV05, Veg96, FK96, AKK99, AVKK02, BVK03}}{\cite{AVKK02}}. In Ref.~\cite{GK11} Gharibian and Kempe generalized these results to the quantum case showing that one also has a PTAS for estimating the energy of dense quantum Hamiltonians over product state assignments. They also asked whether one could improve the result to obtain a PTAS to the ground-state energy \cite{GK11}. We answer this question in the affirmative:

\def\Dense{
For every $\varepsilon > 0$ there is a polynomial-time algorithm such that, given a dense $k$-local Hamiltonian $H$, computes a number $x$ such that $e_{0}(H) \leq x \leq e_0(H) + \varepsilon$.
}

\begin{thm} \label{denseham}
\Dense
\end{thm}

The idea behind the proof of Corollary \ref{denseham} is to show that a product state assignment gives a good approximation to the ground-state energy of dense local Hamiltonians. The statement then follows from the result of Ref.~\cite{GK11}. \thmref{basic} indeed shows that product states give a good approximation in the case of 2-local Hamiltonians. The necessary work here is then to generalize the argument to $k$-local Hamiltonians for arbitrary $k$. 

One approach might be  to employ the perturbation-theory gadgets of Bravyi, DiVincenzo, Loss, and Terhal \cite{BDLT08}.   However, this approach ruins the density by creating a constant fraction of constant-degree vertices.

Instead we will address the case of $k$-local Hamiltonians directly.   In the process, we will flesh out an interesting result that follows from the same information-theoretic ideas as \thmrefs{basic}{planar} and that might be of independent interest: A new quantum version of the de Finetti Theorem \cite{DF80} that applies to general quantum states, and not only to permutation-symmetric states as the previous known versions \cite{HM76, Sto69, RW89, Wer89, CFS02, KR05, CKMR07, Ren07, Yang06, BCY10, BH-local}. Theorem \ref{denseham} follows directly from this result, which we discuss in more detail in section \ref{sec:deF}. 

\subsubsection{Polynomial-time Approximation Scheme for graphs with low threshold rank by Lasserre Hierarchy}\label{sec:threshold}

Another class of Hamiltonians for which we can show a product-state assignment provides a good approximation to the energy, which can moreover be found in polynomial time, are 2-local Hamiltonians whose interaction graph has low \textit{threshold rank} $\rank_{\lambda}$, defined as the number of eigenvalues with value larger than $\lambda$. For a regular graph with $n$ vertices, the degree $D$ is related to the threshold rank as $1/D \leq \lambda^2 + \rank_{\lambda}/n$. Therefore for $\lambda = o(1)$ and $\rank_{\lambda} = o(n)$, Theorem \ref{thm:basic} shows that product states give a good approximation to the ground energy. However this does not say anything about the difficulty of finding the best product-state configuration. 

In \cite{BRS11} the Lasserre hierarchy was shown to give good approximations of 2-CSPs on graphs with low threshold rank (examples of graphs with small threshold rank include small-set expanders and hypercontractive graphs \cite{BRS11}). The main result of this section is essentially an extension of their result to the quantum case.

\def\LowT{
Let $G$ be a $D$-regular graph, and $H = \E_{(i,j)\sim G} H_{i,j}$ a 2-local
Hamiltonian on $n$ qudits with each $\|H_{i,j}\|\leq 1$.  Given $\eps>0$, let 
$k=\poly(d/\eps)\rank_{\poly(\eps/d)}(G)$ and suppose that $n\geq 8k / \varepsilon$. Then it is
possible to estimate the ground-state energy of $H$ to within additive error
$\eps$ in time $n^{O(k)}$.
If there exists a measurement $\{A_x\}$ (i.e. $A_x\geq 0$ and $\sum_x A_x=I$) and a set of permutations $\pi_{i,j}$ such that each $H_{i,j} = -\sum_x A_x \ot A_{\pi_{i,j}(x)}$ then we can instead take $k$ to be $\poly(1/\eps)\rank_{\poly(1/\eps)}(G)$.
}

\begin{thm} \label{thm:lowT}
\LowT
\end{thm}

In fact the algorithm also gives a product state $\ket\varphi = \ket{\varphi_1} \ot \cdots \ot \ket{\varphi_n}$ that achieves this error; i.e. such that
\be
\bra\varphi H \ket{\varphi} \leq e_0(H) + \eps.
\ee

The proof of \thmref{lowT} begins by using the Lasserre hierarchy for lower-bounding the ground-state energy of a Hamiltonian. Since variational methods provide upper bounds to the ground-state energy, such lower bounds can be useful even as heuristics, without any formal proof about their rate of convergence. Indeed, the SDP we discuss has
been proposed as such a heuristic previously, by \cite{BarthelH12,BaumgratzP12}\longonly{
(building on previous work by \cite{Erdahl78, NY96, YN97, Mazz04, PNA10, NGAPP12})}.
The rough idea is to approximate the minimum energy of an $l$-local
Hamiltonian by minimizing over all locally compatible sets of $k$-body
reduced density matrices (for some $k\geq l$) satisfying a global positive semidefinite
condition.  This optimization scheme is an SDP of size proportional to
$\binom{n}{k} d^k$, and as $k\rar n$, it approaches full
diagonalization of the Hamiltonian \footnote{This can be considered as a special case of a more
sophisticated SDP based on the noncommutative
Positivstellensatz~\cite{PNA10, DLTW08, NGAPP12}, whose general performance we will not analyse.}.  

To produce the explicit product-state approximation, as well as the accuracy guarantee, we will need to use a rounding scheme due to \cite{BRS11}, where the Lasserre hierarchy was similarly used for classical 2-CSPs.  Indeed, the structure of the algorithm and proof of \thmref{lowT} follows \cite{BRS11} closely.  In section \ref{sec:threshold-detail} we give the details of applying the Lasserre 
hierarchy to quantum Hamiltonians and give the proof of \thmref{lowT}.

\subsection{de Finetti Theorems with no Symmetry}\label{sec:deF}

The de Finetti theorem says that the marginal probability distribution on $l$ subsystems of a permutation-symmetric probability distribution on $k \geq l$ subsystems is close to a convex combination of independent and identically distributed (i.i.d.) probability measures \cite{DF80}. It allows us to infer a very particular form for the $l$-partite marginal distribution merely based on a symmetry assumption on the global distribution. Quantum versions of the de Finetti theorem state that a $l$-partite quantum state $\rho_l$ that is a reduced state of a permutation-symmetric state on $k \geq l$ subsystems is close (for $k \gg l$) to a convex combination of i.i.d. quantum states, i.e. $\rho_l \approx \int \mu(d\sigma) \sigma^{\otimes l}$ for a probability measure $\mu$ on quantum states \cite{KR05, CKMR07, Ren07, BH-local, DW12}. Quantum versions of the de Finetti theorem have found widespread applications in quantum information theory. However a restriction of all known de Finetti theorems (quantum and classical) is that they only apply to permutation-symmetric states. Is there any way of formulating a more general version of the theorem that would apply to a larger class of states? Here we give one such possible generalization.

We first show a new classical de Finetti theorem for general probability distributions. 

\def\deFinettiC#1{
Let $p^{X_1\ldots, X_n}$ be a probability distribution on $\Sigma^{n}$, for a finite set $\Sigma$.   Let $\vec i = (i_1,\ldots,i_k), \vec j=(j_1,\ldots,j_m)$ be random disjoint ordered subsets of $[n]$, and let $\vec x = (x_1,\ldots,x_m)$ be distributed according to $p^{X_{\vec j}} := p^{X_{j_1} \ldots X_{j_m}}$.    Define $p_{X_{\vec j} = \vec x}$ to be the conditional distribution resulting from taking $X_{j_1}=x_1$, $\ldots$, $X_{j_m}=x_m$, and define $p^{X_{\vec i}}_{X_{\vec j}=\vec x}$ to be the $X_{i_1},\ldots,X_{i_k}$ marginal of this distribution.

Then for every integer $t \leq n - k$ there is an integer $m \leq t$ such that
\longshort {\be}{\begin{multline}}
\E_{\substack{\vec j=(j_1,\ldots,j_m) \\ 
\vec x=(x_1,\ldots,x_m)}} \E_{\vec i=(i_1,\ldots,i_k)}
\left \| p^{X_{\vec i}}_{X_{\vec j} = \vec x} - 
p^{X_{i_1}}_{X_{\vec j} = \vec x} \ot \cdots \ot p^{X_{i_k}}_{X_{\vec j} = \vec x}\right \|_1^2
\shortonly{\\}\leq 
\frac{2 k^2 \ln |\Sigma| }{t},
\label{cdefinettinosymmetry-#1}
\longshort {\ee}{\end{multline}}
}

\begin{thm} \label{thm:deFinetticlassical}
\deFinettiC{first}
\end{thm}

In words the theorem says that given a probability distribution $p^{X_1\ldots X_n}$ over $n$ variables, after conditioning on the values of at most $t$ variables chosen uniformly at random, the remaining $k$-partite marginal distributions are close to a \textit{product} distribution up to error $O(k^2 \ln |\Sigma|/t)$, on average over the the choice of the $k$-partite marginal distribution, the variables $(X_{j_1}, \ldots, X_{j_m})$ that are conditioned on, and the outcomes $(x_1, \ldots, x_m)$ observed.  

We can recover a version of the standard de Finetti theorem directly from \thmref{deFinetticlassical}. Let $p^{X_1,\ldots,X_n}$ be permutation symmetric. Then setting $t = n - k$ and using convexity of the trace norm and the $x^2$ function we find there is a measure $\mu$ on distributions over $\Sigma$ such that \footnote{Note however we do know de Finetti theorems with a better error term. For instance, Diaconis and Freedman proved that the error in the R.H.S. of Eq. (\ref{definettifromconditioning}) can be improved to $O(k^2/n)$ \cite{DF80}.} 
\begin{equation} \label{definettifromconditioning}
\left \Vert  p^{X_1,\ldots,X_k} - \int \mu(dq) q^{\otimes k} \right \Vert_1 \leq \sqrt{ \frac{2 k^2 \ln |\Sigma|}{n - k}}.
\end{equation}

We now turn to the problem of obtaining a quantum version of \thmref{deFinetticlassical}. It turns out that it is straightforward to obtain one, unlike other quantum generalizations of the standard de Finetti theorem that require substantially more work (see \cite{HM76, Sto69, RW89, Wer89, CFS02, KR05, CKMR07, Ren07, Yang06, BCY10, BH-local, DW12}). The only ingredient beyond \thmref{deFinetticlassical} is the idea of applying a product measurement to a multipartite quantum state in order to turn it into a multipartite probability distribution (this will also be a central idea in the proof of the other theorems). If we want to translate distance bounds on the subsystems of the probability distribution into bounds on the original quantum state, then we will need an {\em informationally complete} measurement.

We model a quantum measurement as a quantum-classical channel $\Lambda : {\cal D}(Q) \rightarrow {\cal D}(X)$,
\begin{equation}
\Lambda(\rho) = \sum_{x=1}^{|X|} \tr(M_x \rho) \ket{x}\bra{x}, 
\end{equation}
with $\{ M_x \}$ forming a POVM (positive-operator valued measure), i.e. $0 \leq M_x $, $\sum_x M_x = I$, and $\{ \ket{x} \}$ forming an orthonormal basis.
A measurement $\Lambda$ is called informationally complete if the map $\Lambda$ is injective; i.e. if any two density matrices can be distinguished by their classical measurement outcomes.  We say that the {\em distortion} of a measurement $\Lambda$ is $\sup_{\xi\neq 0} \|\xi\|_1 / \|\Lambda(\xi)\|_1$.  There are several constructions of informationally complete measurements; for instance Ref.~\cite{KR05} gives a construction for a measurement in ${\cal D}(\mathbb{C}^d)$ with distortion $\sqrt{2} d^3$. 

If $\Lambda$ is informationally complete, then it is straightforward to show that $\Lambda^{\ot k}$ is as well. However, we will need explicit bounds on the associated distortion.
\begin{lem}[Informationally complete measurements] \label{lem:info-complete}
For every positive integer $d$ there exists a measurement $\Lambda$ with $\leq d^8$ outcomes such that for every positive integer $k$ and every traceless operator $\xi$, we have
\be
\frac{1}{(18d)^{k/2}}\|\xi\|_1 \leq \|\Lambda^{\ot k}(\xi)\|_1
\label{eq:info-complete}\ee
\end{lem}

The proof is in \secref{info-complete}.

\def\deFinettiQ#1{
Let $\rho^{Q_1, \ldots, Q_n} \in {\cal D}(Q^{\otimes n})$.   For every integer $t\leq n-k$, there is an integer $m\leq t$ such that the following holds.  Let $\Lambda$ be the measurement from \lemref{info-complete}.   Let $\vec i=(i_1,\ldots,i_k), \vec j =(j_1,\ldots,j_m)$ be random disjoint ordered subsets of $[n]$ and $\vec x=(x_1,\ldots,x_m)$ be the measurement outcomes resulting from applying $\Lambda^{\ot m}$ to $Q_{\vec j} := Q_{j_1} \ot \cdots \ot Q_{j_m}$.  Let $\rho_{\vec j, \vec x}$ be the post-measurement state, conditioned on obtaining outcomes $\vec x$.  Then
\longshort {\be}{\begin{multline}} \label{eq:nosym-#1}
\E_{\substack{\vec j = (j_1,\ldots,j_m)\\ \vec x = (x_1,\ldots,x_m)}}
\E_{\vec i = (i_1,\ldots,i_k)}
\left \| \rho^{Q_{i_1}\ldots Q_{i_k}}_{\vec j, \vec x} - 
\rho^{Q_1}_{\vec j, \vec x} \ot \cdots \ot \rho^{Q_k}_{\vec j, \vec x}
\right\|_1^2 \shortonly{\\}\leq 
\frac{4 \ln(d) (18d)^{k} k^2 }{t},
\longshort {\ee}{\end{multline}}
}

\begin{thm} \label{thm:deFinettiquantum}
\deFinettiQ{first}
\end{thm}

\longonly{We note that a variant of \thmref{deFinettiquantum} had been conjectured by Gharibian, Kempe and Regev also in the context of giving a polynomial-time approximation scheme for the local Hamiltonian problem on dense instances \cite{GK12b}. 

It is instructive to compare \thmref{deFinettiquantum} with the usual previous method for applying quantum de Finetti theorems to quantum states without any symmetry.  For example, in quantum key distribution (QKD), the sender and receiver treat each of the $n$ transmitted qubits identically, but nothing forces the eavesdropper to use a symmetric attack.  However, by randomly selecting a subset of $k$ qubits, or equivalently, randomly permuting the $n$ qubits and selecting the first $k$, we can obtain a state which is a reduction of a $n$-partite permutation-symmetric state. This strategy, together with the de Finetti theorem of Ref.~\cite{CKMR07}, yields the bound
\be \left\|\E_{i_1,\ldots,i_k}
\rho^{Q_{i_1}\ldots Q_{i_k}} - 
\int \mu(\text{d}\sigma)\sigma^{\ot k}
\right\|_1 \leq \frac{2d^2k}{n-k}
\label{eq:previous-deF}\ee
Let $\vec j = [n] \backslash \vec i$, denote the $n-k$ systems that are discarded.  Then we can replace the partial trace on $Q_{j_1},\ldots,Q_{j_m}$ with measurement (as indeed is performed by the proof of \cite{CKMR07}) to rewrite Eq. \eq{previous-deF} as
\be \left\|\E_{\vec i, \vec j, \vec x}
\rho^{Q_{i_1}\ldots Q_{i_k}}_{\vec j, \vec x} - 
\int \mu(\text{d}\sigma)\sigma^{\ot k}
\right\|_1 \leq \frac{2d^2k}{n-k}
\label{eq:previous-deF2}\ee

We can thus see that the advantage of Eq. \eq{nosym-first} is that it allows the expectation over $\vec i, \vec j, \vec x$ to be moved outside of the norm.  This is what allows us to control the error when dealing with local Hamiltonians in which the interaction terms are not all identical.  However, the cost of doing so blows up the error by a factor of $d^k$. } It is an interesting question whether Theorem~\ref{thm:deFinettiquantum} can be improved to give a bound that only depends on $\poly(d, k)$, and not on $d^k$. 

For an approximation guarantee based on trace norm, there are examples~\cite[Section II.C]{CKMR07} showing that the error cannot be lower than $\poly(d,k,1/n)$.  However, we can also measure the distance to product using other norms.  One example is the LO (local-operations) norm, which is defined on a $k$-partite system as the maximum distinguishability achievable by local measurements and one-way communication to a referee; i.e.
$$\frac{1}{2}\|\xi\|_{\text{LO}} := \max \{ \sum_{x_1,\ldots,x_k} \lambda_{x_1,\ldots,x_k} \tr(\xi(M_{x_1}^1 \ot \cdots \ot M_{x_k}^k)) : 0\leq \lambda_{x_1,\ldots,x_k} \leq 1, \sum_x M_x^j= I.\}$$
Can we replace the 1-norm on the LHS of Eq. \eq{nosym-first} with the LO norm and remove the $(18d)^k$ term on the RHS of Eq. \eq{nosym-first}?  Such a result would be a simultaneous generalization of \thmref{deFinettiquantum} above and Theorem 2 of our companion paper \cite{BH-local}.

\section{Quantum information theory background}\label{sec:background}

In this section we review some necessary tools from quantum information theory.  State tomography (i.e. the use of an informationally complete measurement) is discussed in \secref{info-complete}.  The use of quantum and classical entropies to quantify correlation is discussed in \secref{mutual-info}.

\subsection{Informationally complete measurement}\label{sec:info-complete}

We prove \lemref{info-complete}, which states that there exists a measurement $\Lambda$ on $d$ dimensions with $\leq d^8$ outcomes such that $\Lambda^{\ot k}$ has distortion $\leq (18d)^{k/2}$.

\begin{proof}
Suppose $\{p_x, \ket{\varphi_x}\}$ is a 4-design, meaning that $\{p_x\}$ is a probability distribution, each $\ket{\varphi_x}$ is a pure state and $\sum_x p_x \varphi_x^{\ot 4}$
is the maximally mixed state on the symmetric subspace of $(\bbC^d)^{\ot 4}$ (see \cite{Har-sym} for background on the symmetric subspace and on $t$-designs).  Eq. (4) of \cite{LW12} (see also Thm 9 of \cite{Montanaro12}) proved that if $\Lambda(\xi) := \sum_x d \,p_x \bra{\varphi_x} \xi \ket{\varphi_x} \proj x$ then Eq. \eq{info-complete} holds.  It remains only to show that 4-designs exist with size $\leq d^8$.  But this follows from Carath\'eodory's theorem, since the space of traceless Hermitian matrices 
on the symmetric subspace has dimension $\binom{d+3}{4}^2-1 \leq d^8-1$.
\end{proof}

\subsection{Multipartite mutual information}\label{sec:mutual-info}

For a state $\rho^{Q_1\ldots Q_k}$ we define the multipartite mutual information as
\begin{multline} \label{cmirelent}
I(Q_1:\ldots :Q_k) := S(\rho^{Q_1\ldots Q_k} || \rho^{Q_1} \otimes \cdots \otimes \rho^{Q_k}) 
\shortonly{\\}= S(Q_1) + \ldots  + S(Q_k) - S(Q_1\ldots Q_k).
\end{multline}
For a quantum-classical state $\rho^{Q_1\ldots Q_kR} = \sum_i p_i \rho_i^{Q_1\ldots Q_k} \otimes \ket{i}\bra{i}_R$ we define the conditional multipartite mutual information as follows
\be \label{CMIdef}
I(Q_1:\ldots :Q_k|R)_{\rho} := \sum_i p_i I(Q_1 : \ldots  : Q_k)_{\rho_i}.
\ee

The multipartite mutual information satisfies the following properties:
\begin{lem} \label{lem:MI} \mbox{}
\benum 
\item \text{Chain Rule:} 
\longshort{\be}{$} I(A:BR) = I(A:R) + I(A:B|R).
\longshort{\label{eq:chainrule}\ee}{$}
\item \text{Multipartite-to-Bipartite \cite{YHHHOS09}:} 
\longshort{\be}{$} I(Q_1:\ldots :Q_k|R) 
= I(Q_1:Q_2|R) \longonly{+ I(Q_1Q_2 : Q_3|R)} + \ldots  + I(Q_1\ldots Q_{k-1} : Q_k | R).
\longshort{\label{eq:multitobi}\ee}{$}
\item \text{Monotonicity under Local Operations:} If 
 $\pi^{Q_1\ldots Q_kR} = \Lambda^{Q_1} \otimes \id^{Q_2\ldots Q_k} (\rho^{Q_1\ldots Q_k R})$, then
\longshort{\begin{equation} \label{monotonicitymulti}}{$}
I(Q_1:\ldots: Q_k | R)_{\pi} \leq I(Q_1:\ldots :Q_k | R)_{\rho}
\longshort{\end{equation}}{$}
\item \text{QC monotonicity:} For any measurements $\Lambda_i^{Q_i\ra X}$ and any state $\rho^{Q_1Q_2}$, we have:
\be I(X_1:X_2)_{(\Lambda_1^{Q_1\ra X_1} \ot \Lambda_2^{Q_2\ra X_2})(\rho)} \leq
I(Q_1:X_2)_{(\id \ot \Lambda_2^{Q_2\ra X_2})(\rho)} \leq
S(Q_1)_\rho.\label{eq:QC-monotoncity}\ee
\item \text{Pinsker's Inequality:} 
\begin{equation} \label{eq:pinsker}
I(Q_1 : \ldots  :Q_k)_{\rho} \geq \frac{1}{2} \Vert \rho^{Q_1\ldots Q_k} - \rho^{Q_1} \otimes \ldots  \otimes \rho^{Q_k} \Vert_1^2.
\end{equation}
\item \text{Upper Limit:} For any state $\rho^{Q_1Q_2}$ and measurements $\Lambda_1,\Lambda_2$,
\ba \label{upperlimit}
I(Q_1: Q_2)_\rho & \leq 2 \min \left( \ln |Q_1|, \ln |Q_2| \right). \\
I(X_1: X_2) _{(\Lambda_1^{Q_1\ra X_1} \ot \Lambda_2^{Q_2\ra X_2})(\rho)} & \leq \min \left( \ln |Q_1|, \ln |Q_2| \right).
\label{eq:QC-upperlimit}
\ea
\eenum
\end{lem}

\section{Self-Decoupling Lemmas}\label{sec:self-decoupling}

The main technical tool in this paper is an information-theoretic bound that we call ``self-decoupling'' and which is based on a very similar technique in \cite{RT12}.  Decoupling refers to the situation when conditioning on one random variable leaves some others nearly decoupled (i.e. independent).  This can be quantified by saying that the conditional mutual information is low.  We use the term ``self-decoupling'' to mean that a collection of random variables can be approximately decoupled by conditioning on a small number of those variables.  We learned about this technique from \cite{RT12}, which introduced it in order to analyze SDP hierarchies.  Our companion paper~\cite{BH-local} applied the same ideas to prove new quantum de Finetti theorems.   

For a distribution $\mu$ on $[n]$, define the distribution $\mu^{\wedge m}$ to be the distribution on $[n]^m$ obtained by sampling $m$ times without replacement according to $\mu$; i.e.
\be \mu^{\wedge m}(i_1,\ldots,i_m) = 
\begin{cases}
0 & \text{if $i_1,\ldots,i_m$ are not all distinct} \\
\frac{\mu(i_1)\cdots \mu(i_m)}{\sum_{j_1,\ldots,j_m\text{ distinct}} \mu(j_1)\cdots \mu(j_m)}
& \text{otherwise}
\end{cases}
\ee
Note that $\mu^{\wedge m}$ is only a valid probability distribution if $m\leq |\supp \mu|$.

\begin{lem}[Self-Decoupling Lemma]\label{lem:self-decoupling}
Let $X_1,\ldots,X_n$ be classical random variables with some arbitrary joint distribution, let $\mu$ be a distribution on $[n]$ and $k < |\supp \mu|$.  Then
\be
\E_{0\leq k' < k}
\E_{(a,b,c_1,\ldots,c_{k'}) \sim \mu^{\wedge k'+2}}
I(X_a : X_b | X_{c_1}\ldots X_{c_{k'}}) \leq
\frac{1}{k} \E_{i\sim \mu} I(X_i : X_{-i}),
\ee
where $X_{-i}:=X_{[n]\backslash \{i\}}$.
\end{lem}

\begin{proof}
Sample $(i,j_1,\ldots,j_k) \sim \mu^{\wedge k+1}$.  Then
\ba
\frac{1}{k} \E_{i\sim \mu} I(X_i : X_{-i}) 
& \geq 
\frac{1}{k} \E_{(i,j_1,\ldots,j_k) \sim \mu^{\wedge k+1}} I(X_i : X_{j_1}\ldots X_{j_k})
& \text{monotonicity} \\
& = 
 \E_{(i,j_1,\ldots,j_k) \sim \mu^{\wedge k+1}}
\E_{0\leq k' < k} I(X_i : X_{j_{k'+1}} | X_{j_1} \ldots X_{j_{k'}})
& \text{chain rule} \\
& = 
\E_{0\leq k' < k}
 \E_{(a,b,c_1,\ldots,c_{k'}) \sim \mu^{\wedge k'+2}}
I(X_a : X_{b} | X_{c_1} \ldots X_{c_{k'}})
\ea
In the last line we have relabeled $i\rar a$, $j_{k'+1} \rar b$ and $j_1,\ldots,j_{k'} \rar c_1,\ldots,c_{k'}$.
\end{proof}

We will actually use a variant of \lemref{self-decoupling} in which there is an additional random variable $Z$ and all entropic quantities are conditioned on $Z$.  Of course this changes nothing in the proof.

We can also derandomize \lemref{self-decoupling} to obtain:
\begin{cor}[Derandomized Self-Decoupling]\label{cor:self-decoupling}
Let $X_1,\ldots,X_n$ be classical random variables with joint distribution $p$, let $\mu$ be a distribution on $[n]$ and $k < |\supp \mu|$.  Then there exists $k'<k$ and $c_1,\ldots,c_{k'} \in [n]$ all distinct such that
\be
\E_{\substack{(a,b) \sim \mu^{\wedge 2} \\ a,b \not\in\{c_1,\ldots,c_{k'}\}}}
I(X_a : X_b | X_{c_1},\ldots, X_{c_{k'}}) \leq
\frac{1}{k} \E_{i\sim \mu} I(X_i : X_{-i}).
\ee
\end{cor}
The proof is immediate.


In \lemref{self-decoupling}, we needed $b,c_1,\ldots,c_{k}$ to be drawn from an exchangeable distribution, but did not make use of the fact that $a$ was also drawn from the same distribution.  Thus, it is possible to generalize the argument to the case where $a$ and $b,c_1,\ldots,c_k$ index entirely different sets of random variables.  A further generalization is to add a random variable $Z$ upon which we condition all of the other variables, as mentioned above.

\begin{lem}[Bipartite Self-Decoupling Lemma]\label{lem:bipartite}
Let $X_1,\ldots,X_m, Y_1,\ldots, Y_n,Z$ be jointly distributed random variables.  Let $\sigma$ be a probability distribution over another random variable $W$.  For each value $w$ in the support of $W$, define distributions $\mu_w$ on $[m]$ and $\nu_w$ on $[n]$.  Choose $k< \min_w |\supp \nu_w|$.
\be
\E_{0\leq k' < k}
\E_{w \sim \sigma}
\E_{\substack{a\sim \mu_w \\ (b,c_1,\ldots,c_{k'}) \sim \nu_w^{\wedge k'+1}}}
I(X_a : Y_b | Y_{c_1}\ldots Y_{c_{k'}},Z) \leq
\frac 1 k 
\E_{w \sim \sigma}
\E_{i \sim \mu_w}
I(X_i : \cup_{j \in \supp \nu_w} Y_j|Z)
\ee
\end{lem}

\begin{proof}
\ba \frac 1 k 
\E_{w \sim \sigma}
\E_{i \sim \mu_w}
I(X_i : \cup_{j \in \supp \nu_w} Y_j|Z) &
\geq
\frac 1 k 
\E_{w \sim \sigma}
\E_{i \sim \mu_w} \E_{j_1,\ldots,j_k \sim \nu_w^{\wedge k}}
I(X_i : Y_{j_1}\ldots Y_{j_k}|Z) \label{eq:bipartite-pf1}\\
& = 
\E_{w \sim \sigma}
\E_{i \sim \mu_w} \E_{j_1,\ldots,j_k \sim \nu_w^{\wedge k}}
\E_{0\leq k' < k}
I(X_i : Y_{j_{k'+1}}|Y_{j_1}\ldots Y_{j_k}, Z) \label{eq:bipartite-pf2} \\
&=\E_{0\leq k' < k}
\E_{w \sim \sigma}
\E_{\substack{a\sim \mu_w \\ (b,c_1,\ldots,c_{k'}) \sim \nu_w^{\wedge k'+1}}}
I(X_a : Y_b | Y_{c_1}\ldots Y_{c_{k'}},Z) \label{eq:bipartite-pf3}
\ea
As in the proof of \lemref{self-decoupling}, \eq{bipartite-pf1}, \eq{bipartite-pf2} and \eq{bipartite-pf3} follow in turn from monotonicity, the chain rule and relabeling variables.
\end{proof}

The main example where we will use \lemref{bipartite} is to analyze the case when $Y_1,\ldots,Y_n$ are variables that have been divided into blocks. Suppose we have partitioned $[n]$ into $n/m$ blocks $V_1,\ldots,V_{n/m}$ with $V_i = \{v_{i,1}, \ldots, v_{i,m}\}$.  Then define the ``coarse-grained'' variables $X_i = (Y_{v_{i,1}}, \ldots, Y_{v_{i,m}})$.  Suppose further that we have a distribution $\nu$ over $[n]$ with an induced distribution over blocks $\mu$; i.e. $\mu(i) = \sum_{a\in V_i} \nu(a)$.  Then we can take $\sigma$ to be a distribution over blocks $w$ and define $\mu_w$ to be $\mu$ conditioned on not choosing $w$ and $\nu_w$ to be $\nu$ conditioned on restricting to $j\in V_w$; i.e.
\be \mu_w(i) = \begin{cases}
0 & \text{if $i=w$} \\
\frac{\mu(i)}{\sum_{i'\neq w}\mu(i')} & \text{otherwise}
\end{cases}
\qquad
 \nu_w(a) = \begin{cases}
0 & \text{if $a\not \in V_w$}\\
\frac{\nu(a)}{\sum_{a'\in V_w}\nu(a')} & \text{otherwise}
\end{cases}
\label{eq:block-correlation}\ee
This will be useful in our proof of \thmref{clustered} where we are interested in the correlations between an individual system and an adjacent block of systems.  For this application we will present a derandomized version of \lemref{bipartite} specialized to this block scenario.

\begin{cor}[Derandomized Block Self-Decoupling]\label{cor:block}
Let $n,m,V_1,\ldots,V_{n/m}, X_1,\ldots,X_{n/m}, Y_1,\ldots, Y_n,Z$ be defined as in 
the above paragraph.  Let $B(a)$ denote the identity of the block containing $a$.  Then for any $k<m$ there exists $k'<k$ and sites $C := \{c^i_j\}_{i\in[n/m], j\in [k']}$ such that $c^i_j\in V_i$ and
\be
\E_{i\in [n/m]}\E_{a\not\in (V_i\cup C)}
I(X_i : Y_a | Y_{c^{B(a)}_1},\ldots,Y_{c^{B(a)}_{k'}},Z) \leq
\frac 1 k \E_{i\neq j} I(X_i : X_j|Z)\label{eq:block-bound}\ee
\end{cor}
\begin{proof}
Invoking \lemref{bipartite} with $\sigma,\mu_w,\nu_w$ chosen as in Eq. \eq{block-correlation} yields
\be
\E_{0\leq k' < k}
\E_{i\in [n/m]}\E_{a\not\in (V_i\cup C)}
\E_{\substack{c_1,\ldots,c_{k'}\\B(a) = B(c_1) = \ldots = B(c_{k'})}} 
I(X_i : Y_a | Y_{c_1},\ldots,Y_{c_{k'}},Z) \leq
\frac 1 k \E_{i\neq j} I(X_i : X_j|Z)\ee
For some choice of $k'$ (which we fix from now on), we have
\be
\E_{i\in [n/m]}\E_{a\not\in (V_i\cup C)}
\E_{\substack{c_1,\ldots,c_{k'}\\B(a) = B(c_1) = \ldots = B(c_{k'})}}
I(X_i : Y_a  | Y_{c_1},\ldots,Y_{c_{k'}},Z) \leq
\frac 1 k \E_{i\neq j} I(X_i : X_j|Z)\ee

Now, for each block $j\neq i$, choose $c_1,\ldots,c_{k'} \in V_i$ to minimize 
$\E_{i\in [n/m]}\E_{a\in V_j} I(X_i:Y_a | Y_{c_1},\ldots,Y_{c_{k'}},Z)$.  Call this choice $c^j_1,\ldots,c^j_{k'}$.
Again using the fact that the minimum of random variable is never greater than its expectation, we obtain Eq. \eq{block-bound}.
\end{proof}

{\em Remark: } 
Can we generalize these results to quantum self-decoupling theorems?
 The chain rule of conditional mutual information works just as well for von Neumann entropies.  So essentially the same argument (modified only by the fact that $I(A:B)$ can be as large as $2\ln |A|$) works there, and was used already in \cite{BCY10}.  However, what fails is the conditioning.  Unlike for classical random variables, the quantum conditional mutual information cannot be written as an average of mutual informations of some conditional distribution.   
This raises a beautiful open problem, which is to understand the structure of quantum states that have very small conditional mutual information~\cite{HJPW04, ILW08, BCY10}.  

The strategy of our paper will instead be to apply the self-decoupling lemmas to the classical random variables resulting from measuring quantum states.  Our companion paper \cite{BH-local} followed a similar strategy, sometimes leaving a system quantum if there was no chance of conditioning on it.  In the language of this paper, this can be thought of as observing that in bipartite self-decoupling (\lemref{bipartite}), we could have taken $X_1,\ldots,X_m$ to be quantum systems.

\section{Proof of main results: Theorems \ref{thm:basic}, \ref{thm:clustered} and \ref{thm:weighted}}\label{sec:main-proof}
This section uses the background results from \secref{background} and the decoupling lemmas from \secref{self-decoupling} to prove Theorems \ref{thm:clustered} (in \secref{proof-clustered}) and \ref{thm:weighted} (in \secref{proof-weighted}).    To avoid redundancy, we give only a sketch of the proof of \thmref{basic} in \secref{proof-basic}.

\subsection{Regular partitioned graphs}\label{sec:proof-clustered}
In this section we complete the proof of Theorem~\ref{thm:clustered}.
\begin{repthm}{thm:clustered}
\ThmClustered{second}
\end{repthm}

\begin{proof}
Let $p^{Y_1,\dots, Y_n} := \Lambda^{\otimes n}(\rho)$ be a probability distribution obtained by measuring all $n$ subsystems of $\rho$ with the informationally complete measurement $\Lambda$ given by Lemma \ref{lem:info-complete}.  Our strategy for this proof is to first construct a variable $Z$ (depending on some of the $Y_i$) such that $\E_{(i,j)\in E} I(Y_i:Y_j|Z)$ is small.  Then we will use this to construct a separable state $\sigma$ for which $\sigma^{Q_iQ_j}$ is close to $\rho^{Q_iQ_j}$ for most $(i,j)\in E$.

We begin by introducing several variables.  Let $X_1,\ldots, X_{n/m}$ be the random variables associated with the partition $V_1\cup \ldots \cup V_{n/m}$ of the sites; i.e. if $V_i = \{v_{i,1}, \ldots,v_{i,m}\}$, then $X_i := (Y_{v_{i,1}},\ldots,Y_{v_{i,m}})$.   Let $\tilde k = \frac{n}{m}\tilde \delta, k=m\delta$ where $\tilde\delta,\delta$ are parameters we will choose later; for now think of them as small positive constants.  Choose $\tilde k'$ uniformly at random from $\{0,1,\ldots,\tilde k-1\}$ and choose $k_1',\ldots,k_{n/m}'$ independently and uniformly at random from $\{0,1,\ldots, k-1\}$.  Choose $\tilde c_1,\ldots,\tilde c_{\tilde k'}$ from $U_{[n/m]}^{\wedge \tilde k'}$ where $U_{[n/m]}$ denotes the uniform distribution on $[n/m]$.  Likewise choose $c^i_1,\ldots,c^i_{k_i'}$ from $U_{V_i}^{\wedge k_i'}$ for each $i\in [n/m]$.  We abbreviate $\tilde c = (\tilde k', \tilde c_1,\ldots,\tilde c_{\tilde k'})$ and $c = (k_1',\ldots,k_{n/m}',\{c^i_a\}_{i\in[n/m], 0\leq a < k_i'})$.  Define also $c_{-j}$ to be $c$ with information about block $j$ deleted; i.e.
\be c_{-j} = 
(k_1',\ldots,k_{j-1}', k_{j+1}',\ldots, k_{n/m}',\{c^i_a\}_{\substack{i\in[n/m]\backslash j\\ 0\leq a < k_i'}}).\ee
Further let $\vec X_{\tilde c}, \vec Y_c, \vec Y_{c_{-j}}$ denote the collections
\subeq{blockequations}{
\vec X_{\tilde c} &= X_{\tilde c_1} \ldots X_{\tilde c_{\tilde k'}}\\
\vec Y_c &= Y_{c^1_1}\ldots Y_{c^{n/m}_{k_{n/m}'}} \\
\vec Y_{c_{-j}} &= Y_{c^1_1}\ldots Y_{c^{j-1}_{k_{j-1}'}} Y_{c^{j+1}_1}
\ldots Y_{c^{n/m}_{k_{n/m}'}}
}
Suppose that we measure the systems in $\vec X_{\tilde c}, \vec Y_c$ (which may overlap) and obtain outcomes $x_{\tilde c}, y_c$ respectively.  Define a random variable $Z$ that takes value $(x_{\tilde c}, y_c)$ and 
denote the resulting post-measurement state by $\tau_z$ (with an implicit $\tilde c, c$ dependence).

Our proof will now proceed by upper-bounding a series of different correlations, on the way to our final statement. The proof proceeds by bounding the following quantities in turn:
\benum[itemsep=-1mm]
\item Classical block-block correlations.
\item Classical block-site correlations.
\item Classical block-site correlations along edges between blocks.  (All later correlations will also be measured in terms of their average across edges between blocks.)
\item Classical site-site correlations.
\item Classical site-site correlations using variational distance.
\item Quantum site-site correlations using trace distance.
\item Constructing the separable approximating state.
\eenum

{\em 1. Classical block-block correlations.}
Our goal is to show that most pairs $X_i,X_{j}$ are approximately independent after being conditioned on a suitably small number of other variables.    \lemref{self-decoupling} suggests that these should be a random collection of $\tilde k'$ other blocks (specifically $\vec X_{\tilde c}$).  Additionally we will condition on $\vec Y_{c_{-i}}$ for reasons that will later be clear.  Now by \lemref{self-decoupling}, we have
 \be 
\E_{\tilde c,c}
\E_{\substack{(i,j) \sim U_{[n/m]}^{\wedge 2} \\ i,j \not\in\{\tilde c_1,\ldots,\tilde c_{\tilde k'}\}}}
I(X_i : X_{j} | \vec X_{\tilde c}, \vec Y_{c_{-i}}) 
\leq 
\E_{\tilde c,c}\frac{m}{n\tilde\delta } \E_{i} I(X_i : X_{-i} | \vec Y_{c_{-i}})_p.\ee
Since $\vec Y_{c_{-i}}$ is a deterministic function of $X_{-i}$ we have 
\be \E_{i} I(X_i : X_{-i} | \vec Y_{c_{-i}})_p 
\leq\E_{i} I(X_i : X_{-i})_p.\ee
Next, the monotonicity of mutual information under local operations (specifically under $\Lambda$) implies that $I(X_i : X_{-i})_p \leq I(V_i : V_{-i})_\rho$. 
 (We remark that in \corref{clustered} we have $\rho = \psi_0$ and have replaced $I(V_i : V_{-i})_\rho$ with $S(V_i)_{\psi_0}$, saving a factor of 2.  This uses part 4 of \lemref{MI}.)
Finally we exchange $i,j$ for notational convenience and conclude that
\be 
\E_{\tilde c}\E_{c}
\E_{\substack{(i,j) \sim U_{[n/m]}^{\wedge 2} \\ i,j \not\in\{\tilde c_1,\ldots,\tilde c_{\tilde k'}\}}}
I(X_i : X_{j} | \vec X_{\tilde c}, \vec Y_{c_{-j}})_p 
\leq \frac{m\bar I}{n\tilde \delta} .
\label{eq:block-block-bound}\ee

{\em 2. Classical block-site correlations.}
For a site $a$, let $B(a)$ denote the identity of the block containing $a$. 
 Now we invoke first \corref{block} and then \eq{block-block-bound} to show that 
\begin{subequations}\label{eq:decoupled-blocks}
\ba
\E_{\tilde c}\E_c
\E_{i \in [n/m]} \E_{a\not\in V_i}
I(X_i : Y_a | \vec X_{\tilde c}, \vec Y_c)_p
& \leq
\E_{\tilde c}\E_{c} \frac{1}{m\delta} \E_{i\neq j} 
I(X_i : X_j|\vec X_{\tilde c}, Y_{c_{-j}})_p
\label{eq:block-site-LHS}\\ & \leq
\frac{\bar I}{n\tilde\delta\delta} 
\ea\end{subequations}
Technically we should have also constrained $a\not\in \{c^i_b\}$, but dropping this constraint will only make the LHS smaller.
Fix a choice of $\tilde c, c$ to minimize the LHS of \eq{block-site-LHS} and use this choice for the rest of the proof.  Also recall that $Z = (\vec X_{\tilde c}, \vec Y_c)$ to obtain
\be
\E_{i \in [n/m]} \E_{a\not\in V_i}
I(X_i : Y_a|Z)_p
\leq 
\frac{\bar I}{n\delta_1\delta_2}.
\label{eq:block-to-site}\ee

{\em 3. Classical block-site correlations along edges between blocks.}
Define $\Gamma(i)$ to be the multiset of sites neighboring $V_i$; i.e. $\Gamma(i)$ is the image of the map $(a,b) \mapsto a$ applied to each $(a,b)\in E$ with $b\in V_i$.   For each $i$, $|\Gamma(i)| = mD\Phi_G(V_i)$, and
$$\sum_{i\in [n/m]} |\Gamma(i)| = mD\sum_{i\in [n/m]} \Phi_G(V_i) = nD\bar\Phi_G.$$
On the other hand, the expectation on the LHS of Eq. \eq{block-to-site} is over $\leq (n/m)\cdot n$ pairs of $i,a$, so we can use it to upper-bound the expectation over $i\in [n/m], a\in\Gamma(i)$
\be \E_{i\in [n/m]} \E_{a\in\Gamma(i)}
I(X_i : Y_a|Z)_p
\leq 
\frac{\bar I/m}{\delta_1\delta_2 D \bar\Phi_G}.\ee

{\em 4. Classical site-site correlations.}
For a pair of sites $a,b$ we will define the relation $a\sim b$ similarly; i.e. let $a\sim b$ if $(a,b)\in E$ and $B(a)\neq B(b)$.  By monotonicity of mutual information under partial trace, we have that $I(Y_a:Y_b|Z) \leq I(X_{B(a)} : Y_b|Z)$.  Also note that the uniform distribution over $i\in [n/m]$ is also obtained by choosing $a$ random and setting $i = B(a)$.  Then
\be \E_{a\sim b}
I(Y_a : Y_b | Z)_p
\leq \frac{\bar I/m}{\delta_1\delta_2 D \bar\Phi_G}.
\label{eq:sites-cl-info}\ee

{\em 5. Classical site-site correlations using variational distance.}
Since the (classical) conditional mutual information is an average of mutual informations of conditional distributions, Eq. \eq{sites-cl-info} can be rewritten as
\be \E_{z} \E_{a\sim b}
I(Y_a : Y_b)_{p_z}
\leq \frac{\bar I/m}{\delta_1\delta_2 D \bar\Phi_G}.\ee
where we have defined $p_z := p|_{Z=z}$
Now we can use Pinsker's inequality and convexity of $x\mapsto x^2$ to obtain
\ba
\E_{z \sim p^Z} \E_{a\sim b}
\|p_z^{Y_a Y_b} - p_z^{Y_a} \ot p_z^{Y_b}\|_1^2
& \leq \frac{2\bar I}{m\delta_1\delta_2 D \bar\Phi_G}. \\
\E_{z \sim p^Z} \E_{a\sim b}
\|p_z^{Y_a Y_b} - p_z^{Y_a} \ot p_z^{Y_b}\|_1
& \leq \sqrt{\frac{2\bar I/m}{\delta_1\delta_2 D \bar\Phi_G}}. 
\label{eq:cl-trace-dist}\ea

{\em 6. Quantum site-site correlations using trace distance.}
Define $C$ to be the set of measured sites, i.e. $C := \bigcup_{i\in [k']}V_{\tilde c_i} \cup \{c^i_j : i\in [n/m], j\in[k_i']\}$. 
Recall that $\tau_z^{Q_1\ldots Q_n}$ be the state that results from measuring (using $\Lambda$) only the sites in $C$ and obtaining outcome $z$.  
 Observe that $z$ has distribution $p^C$ and that
\be \tr_C\rho = \E_z \tr_C\tau_z = \sum_z p^C(z) \tr_C\tau_z.
\label{eq:unmeasured-marginal}\ee
Then Eq. \eq{cl-trace-dist} and \lemref{info-complete} imply that
\be 
\E_{z \sim p^Z} \E_{a\sim b}
\|\tau_z^{Q_a Q_b} - \tau_z^{Q_a} \ot \tau_z^{Q_b}\|_1
 \leq 18d\sqrt{\frac{2\bar I/m}{\delta_1\delta_2 D \bar\Phi_G}}. 
 \label{eq:q-trace-dist}\ee
This can already be seen as a useful intermediate result: measuring a small fraction of the systems leaves a state that is close to product across an average edge.


{\em 7. Constructing the separable approximating state.}
Now define $\sigma_z := \tau_z^{V_1} \ot \cdots \ot \tau_z^{V_{n/m}}$; i.e. $\sigma_z$ is the same as $\tau_z$ within blocks, but has no entanglement between blocks.  Let $\sigma_z = \sum_z p^C(z) \sigma$.
Our final goal is to bound 
\be \E_{(a,b)\in E} \|\rho^{Q_aQ_b} - \sigma^{Q_aQ_b}\|_1.
\label{eq:average-dist}\ee
We do so by considering separately three types of edges: those that are across blocks, those that are within blocks but contain a vertex in $C$, and the rest.  Their definitions, and cardinalities are as follows:
\ban
E_{\text{across}} & := \{(a,b) : a \sim b\} = \{ (a,b) \in E : B(a) \neq B(b)\} \\
E_{\text{measured}} & := \{ (a,b)\in E : B(a)=B(b), \{a,b\}\cap C \neq \emptyset \} \\
E_{\text{within}} &:= E - E_{\text{across}} - E_{\text{measured}}
\ean
We can bound their cardinalities (using $|E|=nD/2$) by
\begin{subequations}\label{eq:edge-counts}\ba
|E_{\text{across}}| & = \frac{1}{2}n D \bar\Phi  = \bar\Phi |E|\\
|E_{\text{measured}}| & \leq |C|D \leq nD(\delta_1 + \delta_2) = 2(\delta_1 + \delta_2)|E|.
\ea\end{subequations}
For $(a,b)\in E_{\text{within}}$, we have $\sigma^{Q_aQ_b} = \rho^{Q_aQ_b}$.  Thus the only contributions to Eq. \eq{average-dist} can come from $E_{\text{across}}$ and $E_{\text{measured}}$.  The latter average we can bound simply by \be \E_{(a,b)\in E_{\text{measured}}}
\|\rho^{Q_aQ_b} - \sigma^{Q_aQ_b}\|_1 \leq 2.\ee
For edges across blocks, we recall that $\sigma_z^{Q_aQ_b} = \tau_z^{Q_a} \ot \tau_z^{Q_b}$.  And for edges that do not intersect $C$, we have $\rho^{Q_aQ_b} = \E_z \tau_z^{Q_aQ_b}$ (see Eq. \eq{unmeasured-marginal}).  Thus, we obtain
\ban \E_{(a,b)\in E_{\text{across}}} 
\|\rho^{Q_aQ_b}-\sigma^{Q_aQ_b}\|_1
& =  \E_{(a,b)\in E_{\text{across}}}  
\|\E_z \tau_z^{Q_aQ_b}-\E_z \tau_z^{Q_a}\ot \tau^{Q_b}\|_1 \\
& \leq  \E_{(a,b)\in E_{\text{across}}}   \E_z
\|\tau_z^{Q_aQ_b}-\tau_z^{Q_a}\ot \tau^{Q_b}\|_1
& \text{triangle inequality}\\
& \leq 
18d\sqrt{\frac{2\bar I/m}{\delta_1\delta_2 D \bar\Phi_G}}.
& \text{using Eq. \eq{q-trace-dist}}
\ean
Putting everything together we have
\be
\E_{(a,b)\in E} \|\rho^{Q_aQ_b} - \sigma^{Q_aQ_b}\|_1 \leq
\bar\Phi_G
18d\sqrt{\frac{2\bar I/m}{\delta_1\delta_2 D \bar\Phi_G}}
+ (\delta_1 + \delta_2).\ee
We balance these  terms by setting $\delta_1, \delta_2$ to be the smallest numbers greater than $(18d)^{1/3}(2\bar\Phi_G\bar I/mD)^{1/6}$ for which $\frac n m \delta_1$ and $m\delta_2$ are integers.  Thus we obtain
\be \E_{(a,b)\in E} \|\rho^{Q_aQ_b} - \sigma^{Q_aQ_b}\|_1 \leq
3(18\sqrt{2})^{1/3}
\L(\frac{d^2\bar\Phi_G\bar I/m}{D}\R)^{1/6} + \frac 1 m + \frac m n.\ee
The numerical constant $3(18\sqrt{2})^{1/3} = 8.825\ldots \leq 9$.
\end{proof}

\begin{figure}
\begin{center}  
\longshort{\includegraphics[width=0.6\columnwidth,angle=0]{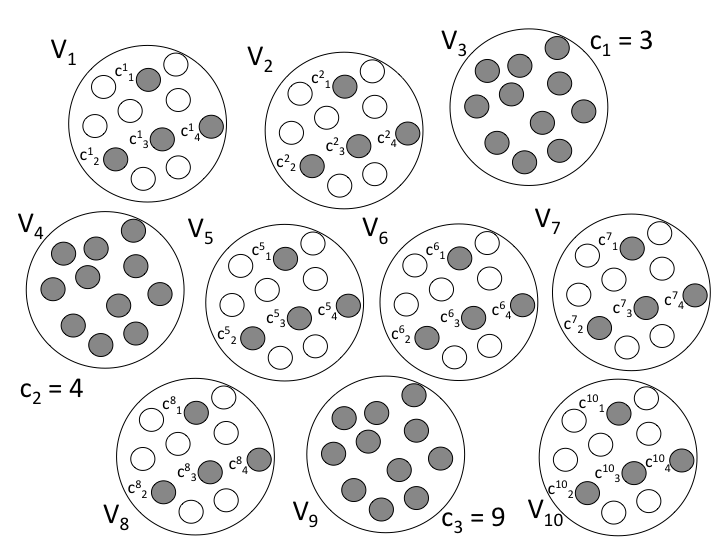}}
{\includegraphics[width=0.9\columnwidth,angle=0]{conditioned.png}}
\caption{
Construction used in the proof of \thmref{clustered}:  
In the first step $k'$ blocks $c_1, \ldots, c_{k'}$ are chosen at random and we condition on the values obtained by measuring $\Lambda$ on all the sites belonging to $V_{c_1} \cup \ldots \cup V_{c_{k'}}$. 
In the second step we condition on the values obtained by measuring $\Lambda$ on sites $\{c^i_j : i\in [n/m], j\in[k'']\}$, where $c^i_j \in V_i$.
In the picture the conditioned sites are denoted by filled disks and the unconditioned sites by unfilled disks. \label{fig00}}  
\end{center}  
\end{figure}

\subsection{Variable degree}\label{sec:proof-weighted}
\begin{repthm}{thm:weighted}
\ThmWeighted{second}
\end{repthm}

\begin{proof}
Let $p^{X_1,\dots, X_n} := \Lambda^{\otimes n}(\rho)$ be a probability distribution obtained by measuring all $n$ subsystems of $\rho$ with the informationally complete measurement   $\Lambda$ given by Lemma \ref{lem:info-complete}. Observe that by Eq. \eq{QC-upperlimit}, $I(X_i:X_{-i})\leq \ln(d)$ for each $i$.
 Let $\delta>0$ be a constant we will choose later.  Applying self-decoupling (\lemref{self-decoupling}) with $\mu=\pi$ and $k$ to be chosen later, we find
\be 
\E_{0\leq k' < k}\E_{(a,b,c_1,\ldots,c_{k'}) \sim \pi^{\wedge k'+2}}
I(X_a : X_b | X_{c_1},\ldots, X_{c_{k'}}) \leq
\frac{\ln(d)}{k}.
\label{eq:weighted-decouple}\ee
Define $C:= \{c_1,\ldots,c_{k'}\}$ and define $\gamma := \E_{a,b} I(X_a:X_b|X_C)_p$, where the expectation is taken over $a,b$ sampled from $\pi^{\wedge 2}$ conditioned on $a,b\not\in C$.  Then Eq. \eq{weighted-decouple} could be expressed succinctly as 
\be \E_C[\gamma] \leq \ln(d)/k.
\label{eq:short-weighted-decouple}\ee
  (Note that $k'$ is determined by $C$, so $\E_C$ represents the average over choices of both $k'$ and $c_1,\ldots,c_{k'}$.)

Abbreviate $Z := X_{c_1}\ldots X_{c_{k'}}$ and define
\be\Delta_z(a,b) := 
\begin{cases}
\|p^{X_a X_b}|_{Z=z} - p^{X_a}|_{Z=z} \ot p^{X_b}|_{Z=z}\|_1
&
\text{if }a\neq b \text{ and }a,b\not\in C \\
0 & \text{otherwise}\end{cases}
\ee
By Pinsker's inequality,
\be \Delta_z(a,b) \leq \sqrt{2 I(X_a : X_b  | Z=z)_p}.
\label{eq:Delta-Pinsker}\ee
Define the inner product $\langle A,B\rangle := \tr(A^\dag  B)$, and define $\Delta_z^2$ to be the entrywise square of $\Delta_z$.
.  Then \eq{Delta-Pinsker} implies that
\be 
 \E_{z\sim p^Z} \langle \pi\pi^T, \Delta_z^{2}\rangle
\leq 2\gamma.
\label{eq:pipi-bound}\ee

We would like to show that $\E_z \Delta_z(a,b)$ is small when we average it over $(a,b)\sim G$.  However, Eq. \eq{pipi-bound} bounds the average over the distribution $\pi\pi^T$.  To relate these, we will divide $G$ into a component that is well approximated by $\pi\pi^T$ and a component that isn't.   For some parameter $\lambda>0$ to be determined later, define
$$B := \{(i,j) : G_{i,j} > \lambda \pi_i\pi_j\} = \{(i,j) : A_{i,j} > \lambda \pi_i\}.$$
We can then bound
\be\eta := 
\sum_{(i,j)\in B} G_{i,j} = \sum_{(i,j)\in B} A_{i,j} \pi_j
\leq \sum_{(i,j)\in B} \frac{A_{i,j} A_{j,i}}{\lambda}
\leq \frac{\tr[A^2]}{\lambda}.\label{eq:eta-bound}\ee
By the definition of $\eta$, we can find a probability distribution $G'$ such that $G \leq \lambda \pi\pi^T + \eta G'$, where $\leq$ is an element-wise inequality.

We will now bound the average correlation across edges.  Define 
$\Delta := \E_{z\sim p^Z}\Delta_z$.    Then
\begin{subequations}\label{eq:G-Delta}\ba \langle G, \Delta \rangle
& \leq \sqrt{ \langle G, \Delta^2 \rangle}
& \text{Cauchy-Schwarz}\\
&\leq  \sqrt{ \langle \lambda\pi\pi^T, \Delta^2 \rangle + 
\langle \eta G', \Delta^2 \rangle} \\
&\leq \sqrt{2\gamma\lambda + 4 \eta} 
& \text{using Eq. \eq{pipi-bound} and $\Delta(a,b)\leq 2$} \\
&\leq \sqrt{2\gamma\lambda + 4\frac{\tr[A^2]}{\lambda}}
& \text{using Eq. \eq{eta-bound}}\\
&= \L(32\gamma\tr[A^2]\R)^{1/4} 
&\text{choosing $\lambda$ to balance the two terms.}
\ea\end{subequations}

Now we turn to approximating $\rho$ with a mixture of product states.  Define $\tau$ to be the state resulting from measuring $\rho$ only on systems $Q_{c_1},\ldots,Q_{c_{k'}}$.  Let $Q'$ denote the remaining $n-k'$ quantum systems.  Then we can write 
$$\tau = \E_{z\sim p^Z} \tau_z^{Q'} \ot \proj{z}^Z,$$
and we have $\rho^{Q'} = \tau^{Q'}  = \E_{z\sim p^Z} \tau_z^{Q'}$.  It will be convenient to extend each $\tau_z$ to be defined on all $Q_1\ldots Q_n$ by defining 
$$\tau_z^{Q_1\ldots Q_n} := \tau_z^{Q'} \ot \bigotimes_{c\in C} (I/d)^{Q_c}.$$
(The state $I/d$ here is an arbitrary choice; any other state would also work for the proof.)  Now define
\be \sigma_z^{Q_1\ldots Q_n} := \tau_z^{Q_1} \ot \ldots \ot \tau_z^{Q_n}, \ee
and $\sigma := \E_{z\sim p^Z} \sigma_z$. Then
\ban
\E_{z\sim p^Z}\E_{(a,b)\sim G} \|\tau_z^{Q_aQ_b} - \sigma_z^{Q_aQ_b}\|_1
&= 
\E_{z\sim p^Z}\E_{(a,b)\sim G} \|\tau_z^{Q_aQ_b} - \tau_z^{Q_a} \ot \tau_z^{Q_b}\|_1 \\
&\leq \E_{z\sim p^Z}\E_{(a,b)\sim G} 18 d \Delta_z(a,b)
&\text{using \lemref{info-complete}}\\
& \leq 
18d \L(32\gamma\tr[A^2]\R)^{1/4} 
& \text{from Eq. \eq{G-Delta}}
\ean
Finally we evaluate the total error.  One new challenge is that $\pi(C)$ is not always $\leq k/n$. 
  Observe that
\ba 
\E_{(a,b)\sim G} \|\rho^{Q_a Q_b} - \sigma^{Q_aQ_b}\|_1 &
\leq 2\pi(C) + \E_{(a,b)\sim G} \Delta(a,b) \\
&\leq 2\pi(C) + 18d \L(32\gamma\tr[A^2]\R)^{1/4} 
\label{eq:C-objective}\ea
Now $\E_C[\pi(C)] = \E_{k'}[k'] \|\pi\|_2^2 \leq \frac{1}{2}k \|\pi\|_2^2$.
Choose $C$ to minimize Eq. \eq{C-objective}.  This yields
\ban 
\E_{(a,b)\sim G} \|\rho^{Q_a Q_b} - \sigma^{Q_aQ_b}\|_1 & \leq \E_C\L[2\pi(C) + 18d \L(32\gamma\tr[A^2]\R)^{1/4}\R] 
&\text{since min $\leq$ expectation}\\ & \leq k\|\pi\|_2^2 + 18d \L(32\E_C[\gamma]\tr[A^2]\R)^{1/4}  
&\text{convexity of $x\mapsto x^4$}
\\ & \leq k\|\pi\|_2^2 + 18d \L(32\frac{\ln(d)\tr[A^2]}{k}\R)^{1/4} 
&\text{from Eq. \eq{short-weighted-decouple}}
\\ & \leq 2 \sqrt{18d} \L(32\ln(d)\tr[A^2]\|\pi\|_2^2\R)^{1/8} + \|\pi\|_2^2
& \text{choosing $k$ optimally}
\\ & \leq 14 \L(d^4\ln(d)\tr[A^2]\|\pi\|_2^2\R)^{1/8} + \|\pi\|_2^2
\ean
In the line where we choose $k$ optimally, we are choosing $k$ to balance the two terms, and then rounding up to the nearest integer, which incurs an extra error of up to $\|\pi\|_2^2$.

To prove Eq. \eq{H-weighted-second}, we choose $z$ to minimize $\tr(H\sigma_z)$, and then further choose a pure product state $\ket{\varphi}$ from the support of $\sigma_z$ to minimize $\tr(H\varphi)$.
\end{proof}

\subsection{Proof of the basic version of the theorem}\label{sec:proof-basic}
Here we sketch the proof of \thmref{basic}, mostly by pointing out places in which the proof of \thmref{weighted} can be modified.

The first change is that $C$ can be immediately chosen to minimize $\gamma$, achieving $\gamma \leq \ln(d)/k$.  Next, we can set $\lambda = D/n$ and have $\eta=0$.  Thus we obtain 
$$\langle G, \Delta\rangle \leq \sqrt{\langle G, \Delta^2\rangle }
\leq \sqrt{\lambda \langle \pi\pi^T, \Delta^2\rangle }
\leq \sqrt{\frac{2n\ln(d)}{kD}}.$$
This classical trace distance is again multiplied by $18d$ to obtain the quantum trace distance.   Measuring $k$ vertices ruins a $2k/n$ fraction of the edges.  The total error is then
$$18d\sqrt{\frac{2n\ln(d)}{kD}} + \frac{2k}{n}.$$
We then choose $k$ to balance these terms, and round up to the nearest integer, incurring a further error that is $\leq 1/n$.  This yields error
$$\L(\frac{1296 d^2\ln(d)}{D}\R)^{1/3} + \frac 1 n.$$
We obtain the claimed bound by using $n\geq D$ and assuming that the first term is $\leq 2$.

\section{Proof of amplification results: \propref{CSPs} and \corref{amplification}} \label{sec:proofCSPs}

\begin{repprop}{prop:CSPs}
\propCSPs
\end{repprop}

\begin{proof}

By the PCP theorem it follows there is a constant $\varepsilon_0 > 0$ such that it is $\NP$-hard to decide whether a 3-$\SAT$(5) formula ${\cal F}$ in which each variable belongs to 5 clauses is satisfiable or if $\text{unsat}({\cal F}) \geq \varepsilon_0$. 

The first step in the proof is to map this problem to a label cover problem (which is a 2-CSP). Here we follow Ref. \cite{Sal11} . A label cover problem is specified by the following parameters: 

\benum
\item a bipartite graph $G(V \cup W, E)$;
\item labels $[N]$ and $[M]$ for the vertices in $V$ and $W$, respectively;
\item a function $\{ \Pi_{v, w} \}_{(v, w) \in E}$ on every edge $(v, w) \in E$ such that $\Pi_{v, w} : [M] \rightarrow [N]$.
\eenum

A labelling $l : V \rightarrow [N], W \rightarrow [M]$ covers the edge $(v, w)$ if $\Pi_{v, w}(l(w)) = l(v)$. The goal of the problem is to find a labelling of the vertices that covers the maximum number of edges.

The reduction from 3-$\SAT$(5) with variables $\{ x_1, \ldots, x_n \}$ and constraints $\{ C_1, \ldots, C_m \}$ to label cover is as follows. Let $V = \{x_1, \ldots, x_n \}$ be the set of left vertices and $W = \{ C_1, \ldots, C_m\}$ the set of right vertices. Two vertices are adjacent if, and only if, $x_i \in C_j$. The constraint $\Pi_{x_i, C_j} : [7] \rightarrow [2]$ takes one of the possible 7 satisfying assignments for $C_j$ and outputs the value of $x_i$ in the assignment. Let's denote this label cover problem by ${\cal L}$. If ${\cal F}$ is satisfiable then so is ${\cal L}$. On the other hand, if $\text{unsat}({\cal F}) \geq \varepsilon_0$, then $\text{unsat}({\cal L}) \geq \varepsilon_0/3$ (see e.g. \cite{Sal11}). Therefore it is $\NP$-hard to determine whether $\text{unsat}({\cal L}) = 0$ or $\text{unsat}({\cal L}) \geq \varepsilon_0/3$. 

Note that in ${\cal L}$ the left vertices have degree $5$ and the right vertices have degree $3$, while the alphabet size $|\Sigma| = 7$. The next step is to increase the degree while making the graph regular. We proceed as follows: we create new variables $\{ \tilde x_{i} \}$ as left vertices and to each of the right vertices $C_j(x_{i_1}, x_{i_2}, x_{i_3})$ we associate $s = \lceil|\Sigma|^{2 \alpha/\beta}\rceil - 3$ different $\tilde x_{i}$ and add trivial constraints that are always satisfied between $C_j$ and each of them, i.e. we replace $C_j(x_{i_1}, x_{i_2}, x_{i_3})$ by $C_j'(x_{i_1}, x_{i_2}, x_{i_3}, \tilde{x}_{k_1}, \ldots, \tilde{x}_{k_s})$ $= C_j(x_{i_1}, x_{i_2}, x_{i_3})$. In the process we also make sure that each of the new variables $\tilde{x}_j$ is connected to $|\Sigma|^{2 \alpha/\beta} $ different right vertices. Then we create new constraints $\tilde C_j$ and to each of the left vertices $x_i$ we associate $|\Sigma|^{2 \alpha/\beta} - 5$ of them, while making sure that each of the new constraints $\tilde C_j$ is also connected to $\lceil |\Sigma|^{2 \alpha/\beta}\rceil$ variables. Denote this new instance by ${\cal C}$. We perform this mapping to make sure that $|\Sigma|^{\alpha} \leq \mathrm{deg}({\cal C})^{\beta}$, with $\text{deg}({\cal C})$ the degree of ${\cal C}$. This procedure increases the total number of constraints by a multiplicative constant factor (depending on $\alpha, \beta$). Therefore $\text{unsat}({\cal C}) \geq \Omega(\text{unsat}({\cal L}) )$.

The next step is to apply parallel repetition to ${\cal C}$. For an integer $t > 0$, consider the following label cover instance ${\cal C}_t$ specified by 

\benum
\item a bipartite graph $G'(V', W', E')$ with $V' = V^t$ and $W' = W^t$; 
\item labels $[N'] = [N]^t$ and $[M'] = [M]^t$ for the vertices in $V'$ and $W'$, respectively;
\item an edge set $E'$ such that $(v', w') \in E'$ if, and only if, $(v_{i_j}, w_{i_j}) \in E$ for all $i$ and $1 \leq j \leq t$, with $v' = \{ v_{i_1}, \ldots, v_{i_t} \}$ and $w' = \{ w_{i_1}, \ldots, w_{i_t} \}$;
\item a function $\{ \Pi_{v', w'} \}_{(v', w') \in E'}$ on every edge $(v', w') \in E'$ such that $\Pi_{v', w'} : [M'] \rightarrow [N']$ with $\Pi_{v', w'}(b_1, \ldots, b_t) = \{ \Pi_{v_{i_1}, w_{i_1}}(b_1), \ldots, \Pi_{v_{i_t}, w_{i_t}}(b_t)  \}$. 
\eenum

We claim that $\text{unsat}({\cal C}_t) \geq \text{unsat}({\cal C})$. Indeed one way to see this fact is by evoking Raz's parallel repetition theorem \cite{Raz98}. However here we only need the easy implication that the unsat is not decreased by parallel repetition (as opposed to the fact that the unsat approaches 1 exponentially in $t$). Moreover the degree of each vertex is larger than $|\Sigma|^{2 t \alpha/\beta}$ and so $\text{deg}_t \geq |\Sigma|^{2 t \alpha/\beta}$, with $\text{deg}_t$ the average degree of ${\cal C}_t$, while the alphabet size of ${\cal C}_t$ is $|\Sigma_t| = |\Sigma|^t$. Therefore for a constant $t$ large enough $\Omega(\varepsilon_0) \geq c |\Sigma|^{-t \alpha} \geq  c |\Sigma_t|^{\alpha}/(\text{deg}_t)^{\beta}$.
\end{proof}

\begin{repcor}{cor:amplification}
\CorAmplification
\end{repcor}

\begin{proof}
Suppose we are given a 2-local Hamiltonian $H$ with ``PCP normalization''; i.e. $H  = \sum_{i,j} G_{i,j} H_{i,j}$, with $G$ a probability distribution (with $G_{i,i}=0$ for each $i$) and $H_{i,j}$ acting on qudits $i,j$ with $\|H_{i,j}\|\leq 1$.  Consider the problem of determining whether $e_0(H)\leq 0$ or $e_0(H)\geq\eps$ subject to the promise that one of these holds.  Pad $H$ with enough dummy edges to bring the degree to $D$, for $D$ a parameter to be determined later.  For each of these ``dummy edges'' the corresponding term in the Hamiltonian is $H_{i,j}=0$.  Let $H'$ be the resulting Hamiltonian.  Note that $H' = cH$ for $1/D \leq c \leq 1$.  Thus our two cases correspond to $e_0(H')\leq 0$ or $\geq \eps/D$.

Now apply $\cP_t$ to $H'$. By assumption, we have either $e_0(\cP_t(H'))\leq 0$ or $e_0(\cP_t(H'))\geq \eps/D$.  However, by \corref{NP}, we can approximate $e_0(\cP_t(H'))$ to accuracy $12 (d^{2t}\ln(d^t) / D^t)^{1/3}$ with a NP witness of size at most $n^{O(t)}$.  (For simplicity, we neglect the additional arithmetic errors from discretization.) Now choose $D = 8d^3$ and use $\ln(d)\leq d$, so that this accuracy is $\leq 12 \cdot 2^{-t}$.  If $t \geq \log_2(12D/\eps) = \log_2(96 d^3/\eps)$, then this will be enough to distinguish our two cases.  Since $d$ and $\eps$ are constants, the witness size and the computation time are both polynomial in the original size of the computation.  Thus, our original energy estimation problem is in \NP, which implies the quantum PCP theorems is false (assuming that $\NP \neq \QMA$.)
\end{proof}

\section{Proof of Theorem \ref{thm:planar}}

\begin{repthm}{thm:planar}
\Planar
\end{repthm}

\begin{proof}

Consider a partition of the vertices of the graph into two sets: ${\cal H}$ and ${\cal L}$;  ${\cal H}$ contain all the vertices with degree larger than $f(\varepsilon) := (23328 d^2 \ln(d) / \varepsilon)^2$, while  ${\cal L}$ contains the remaining vertices.

Following the proof of \thmref{clustered}, let $p^{X_1, \ldots, X_n} := \Lambda^{\otimes n}(\psi_0^{Q_1, \ldots, Q_n})$ be a probability distribution obtained by measuring all $n$ subsystems of $\psi_0$ with the informationally-complete measurement $\Lambda$ given by Lemma \ref{lem:info-complete}.
Then by Lemma \ref{lem:bipartite} and Markov's inequality, there is a $k' \leq k$ such that with probability larger than $2/3$ over the choice of vertices $c_1, \ldots, c_k'$, 
\begin{equation} \label{selfdecouplinguse}
\mathbb{E}_{i \in {\cal H}} \mathbb{E}_{j \neq i} I(X_i : X_j | X_{c_1}, \ldots, X_{c_k'}) \leq \frac{3\ln(d)}{k}.
\end{equation}
Also by Markov's inequality with probability larger than $2/3$,
\begin{equation} \label{boundaveragedegree}
\frac{1}{k'}\sum_{j=1}^{k'} \text{deg}(c_j) \leq 3 \overline{D},
\end{equation}
with $\overline{D}$ the average degree of the graph. Thus by the union bound we can choose $c_1, \ldots, c_k'$ such that both Eqs. (\ref{selfdecouplinguse}) and (\ref{boundaveragedegree}) hold true. 

Let $\tau$ be the state resulting from measuring $\psi_0$ only on systems $Q_{c_1}, \ldots, Q_{c_{k'}}$, and replacing the measured state on those systems with an arbitrary state, such as maximally mixed state. 
We can write $\tau = \mathbb{E}_{z \sim p^Z} \tau_z$, where $\tau_z$ is the state conditioned on obtaining measurement outcome $z$.  Now define $\sigma = \mathbb{E}_{z \sim p^Z} \sigma_{z}^{Q_1, \ldots, Q_n}$, with
\begin{equation} \label{sigmaz}
\sigma_{z}^{Q_1, \ldots, Q_n} := \left( \bigotimes_{i \in H} \tau_z^{Q_i} \right) \otimes \tau_z^{Q_{L}}, 
\end{equation}
where $Q_L := \cup_{i \in L} Q_i$; i.e. we break the entanglement of the sites in the high-degree region $\cH$ and keep the original density matrix in the low-degree region $\cL$. Then
\begin{equation} \label{boundingenergy}
e_{0}(H) \geq \tr(H \tau) - \frac{6k}{n} \geq \tr(H \sigma) - \mathbb{E}_{i \in  {\cal H}} \mathbb{E}_{ (i, j) \sim G} \Vert \tau^{Q_iQ_j} - \sigma^{Q_iQ_j}   \Vert_1   - \frac{6k}{n}, 
\end{equation}
where the first inequality follows from the fact that we measure $k' \leq k$ subsystems and their average degree is at most triple the average degree (Eq. (\ref{boundaveragedegree})), and the second from the fact that $H$ is a 2-local Hamiltonian and that we consider a product ansatz only for the vertices in $ {\cal H}$. 
Now define $\Gamma(i)$ to the neighbors of $i$ and 
use in turn the triangle inequality, \lemref{info-complete}, Pinsker's inequality and convexity to bound
\begin{eqnarray}
\mathbb{E}_{i \in  {\cal H}} \mathbb{E}_{j \sim\Gamma(i)} \Vert \tau^{Q_iQ_j} - \sigma^{Q_iQ_j}   \Vert_1 &\leq& \mathbb{E}_{i \in H} \mathbb{E}_{j \sim\Gamma(i)} \mathbb{E}_{z \sim p^Z}  \Vert \tau_z^{Q_iQ_j} - \sigma_z^{Q_iQ_j}   \Vert_1 \nonumber \\
&\leq&  18d \mathbb{E}_{i \in  {\cal H}} \mathbb{E}_{ j \sim\Gamma(i)}  \sqrt{2 I(X_i : X_j | X_{c_1}, \ldots, X_{c_k'})} \nonumber \\
&\leq& 18d \sqrt{2\mathbb{E}_{i \in  {\cal H}} \mathbb{E}_{j \sim\Gamma(i)}  I(X_i : X_j | X_{c_1}, \ldots, X_{c_k'})} 
\end{eqnarray}
Since $\text{degree}(i) \geq f(\varepsilon)$ for all $i \in  {\cal H}$, choosing $k = \delta n$ and using Eq. (\ref{selfdecouplinguse}) gives
\begin{equation} 
\mathbb{E}_{i \in  {\cal H}} \mathbb{E}_{j \sim\Gamma(i)}   I(X_i : X_j | X_{c_1}, \ldots, X_{c_k'}) \leq \frac{3\ln(d)}{\delta f(\varepsilon)}.
\end{equation}
Thus by Eq. (\ref{boundingenergy}) 
\begin{equation} 
\tr(H \sigma) - e_0(H) \leq 2 \delta + 18d \sqrt{\frac{6\ln(d)}{\delta f(\varepsilon)} }.
\end{equation}
Choosing $\delta$ to balance the terms we get 
\begin{equation} 
\tr(H \sigma) - e_0(H) \leq 18d \left(  \frac{6 \ln(d)}{f(\varepsilon)}  \right)^{1/2},
\end{equation}
and so by Eq. (\ref{sigmaz}) there is a product state $\pi := \bigotimes_{i \in H} \pi_{i} \otimes \pi_{L}$ such that 
\begin{equation} 
\tr(H \pi) - e_0(H) \leq 18d \left(  \frac{6 \ln(d)}{f(\varepsilon)}  \right)^{1/2}.
\end{equation}

Let us now focus on $\pi_ {\cal L}$. Consider the subgraph of $G$ whose edges only act on the region $ {\cal L}$. This is again a planar graph. By section 3.1 of \cite{BBT09} we can partition this subgraph into disjoint graphs each of size at most $f(\varepsilon)^{O(1/\varepsilon)}$ by deleting at most $\varepsilon n / 3$ edges. Let $ {\cal L} = \{  {\cal L}_1, \ldots,  {\cal L}_u \}$ be a partition of $ {\cal L}$ such that each $ {\cal L}_j$ contains the vertices of one of the disjoint subgraphs. Then 
\begin{eqnarray}
\tr\L(H \bigotimes_{i \in  {\cal H}} \pi_{i} \otimes 
\bigotimes_{j\in [u]} \sigma_{ {\cal L}_j}\R) &\leq& \tr(H \pi) + 18d \left(  \frac{6 \ln(d)}{f(\varepsilon)}  \right)^{1/2} +  \varepsilon / 3 \leq \tr(H \pi) + \varepsilon/2
\end{eqnarray}

Therefore it suffices to estimate 
\be
\max_{\psi_{ {\cal H}_1}, \ldots, \psi_{ {\cal H}_k}, \psi_{ {\cal L}_1}, \ldots, \psi_{ {\cal L}_u}} \bra{\psi_{ {\cal H}_1}, \ldots, \psi_{ {\cal H}_k}, \psi_{ {\cal L}_1}, \ldots, \psi_{ {\cal L}_u}} H \ket{\psi_{ {\cal H}_1}, \ldots, \psi_{ {\cal H}_k}, \psi_{ {\cal L}_1}, \ldots, \psi_{ {\cal L}_u}}
\ee
to within error $\varepsilon/2$.

For this goal for each $i \in  {\cal H}$ and $ {\cal L}_j$ we consider a $\delta$-net on $\mathbb{C}^{\tilde d}$, with ${\tilde d}$ the dimension of the space on $i$ or $ {\cal L}_j$. It is given by a set ${\cal T}_{\delta}$ of less than $(5/\delta)^{2{\tilde d}} \leq \exp(O(\ln(1/\delta) \exp(O(f(\varepsilon)^{O(1/\varepsilon)}))))$ states $\ket{\phi_k}$ such that for all $\ket{\psi} \in \mathbb{C}^{\tilde d}$ there is a $k$ such that $\Vert \psi - \phi_k  \Vert_1 \leq \delta$. We have
\ba 
 \min_{ \{\psi_i\}_{i \in \cH}, \{\psi_{\cL_j}\}_{j \in [u]}} \tr (H (\bigotimes_{i\in H} \psi_i \ot \bigotimes_{j \in [u]}\psi_{\cL_j})) 
 \leq 
\min_{ \{\psi_i\in \cT_\delta\}_{i \in \cH}, \{\psi_{\cL_j}\in\cT_\delta\}_{j \in [u]}} \tr (H (\bigotimes_{i\in H} \psi_i \ot \bigotimes_{j \in [u]}\psi_{\cL_j})) - 2\delta 
\ea
Therefore we can solve the problem by performing the optimization only over ${\cal T}_{\delta}$ for a sufficiently small $\delta$. Under this restriction the problem becomes the minimization of the ground-state energy of a classical Hamiltonian (with $|{\cal T}_{\delta}|$ possible values for the spin in each site) over a planar graph and it follows from \cite{BBT09} that one can obtain a $\nu$-additive approximation for this problem in polynomial-time for all $\nu > 0$. Finally we take $\delta = \varepsilon/6$ and $\nu = \varepsilon/6$, so that our total error is $\leq\varepsilon$.
\end{proof}

\section{Proof of Theorem \ref{denseham}}

The next lemma is an easy corollary of a result of Gharibian and Kempe \cite{GK11}, who used it to give a polynomial-time algorithm achieving a $d^{-k+1}$-approximation for the problem of computing the maximum eigenvalue of dense $k$-local Hamiltonians on $d$-dimensional particles. 

\begin{lem} \label{GKlemma}
For all $\varepsilon > 0$ there is a polynomial-time algorithm for dense $k$-local Hamiltonian $H$ that outputs a product state assignment with value at most $\varepsilon$ larger than  
\be
\min_{\psi_1, \ldots, \psi_n} \bra{\psi_1, \ldots, \psi_n}H\ket{\psi_1, \ldots, \psi_n}.
\ee
\end{lem}

\begin{proof}
Given a Hamiltonian $H' = l^{-1}\sum_{j=1}^l Q_j$ with $0 \leq Q_j \leq \id$, Theorem 4 of Ref. \cite{GK11} gives an algorithm for computing a $(1+\varepsilon)$-multiplicative approximation of $\max_{\ket{\psi_1, \ldots, \psi_n}} \bra{\psi_1, \ldots, \psi_n} H' \ket{\psi_1, \ldots, \psi_n}.$ Let $H = l^{-1}\sum_{j=1}^l H_j$, with $\Vert H_j \Vert \leq 1$, and define $H' = I - l^{-1}\sum_{j=1}^l H_j = l^{-1} \sum_{j=1}^l (I - H_j)$. As $\max_{\ket{\psi_1, \ldots, \psi_n}} \bra{\psi_1, \ldots, \psi_n} H \ket{\psi_1, \ldots, \psi_n} = 1 -  \max_{\ket{\psi_1, \ldots, \psi_n}} \bra{\psi_1, \ldots, \psi_n} H' \ket{\psi_1, \ldots, \psi_n}$, we find that we can estimate the L.H.S. within additive error $\varepsilon$ using the  $(1+\varepsilon)$-multiplicative approximation of the second term of the R.H.S..
\end{proof}

We now turn to the main result of this section. 

\begin{repthm}{denseham}
\Dense
\end{repthm}

\begin{proof}

We will prove that there is a product state $\ket{\psi_1} \otimes \ldots \otimes \ket{\psi_n}$ such that
\be \label{errorwithproductstate}
\bra{\psi_1, \ldots, \psi_n} H \ket{\psi_1, \ldots, \psi_n} \leq e_0(H) + O(n^{-1/4}).
\ee
The result then follows from Lemma \ref{GKlemma}.

To prove Eq. (\ref{errorwithproductstate}) we use \thmref{deFinettiquantum}. Let $\rho := \ket{\psi_0}\bra{\psi_0}$ be the projector onto a ground state of $\ket{\psi_0}$. Then, given an integer $t$, we find that there are sites $j_1, \ldots, j_m$ such that 
\begin{eqnarray}  \label{bounddefinettinosymmetryapplied}
&& \mathop{\mathbb{E}}_{a_1\ldots, a_m}  \mathop{\mathbb{E}}_{i_1\ldots, i_k}  \left \Vert \rho_{X_{j_1}=a_1, \ldots, X_{j_m}=a_m}^{Q_{i_1}\ldots, Q_{i_k}}  - \rho_{X_{j_1}=a_1, \ldots, X_{j_m}=a_m}^{Q_{i_1}} \otimes \ldots \otimes  \rho_{X_{j_1}=a_1, \ldots, X_{j_m}=a_m}^{Q_{i_k}}     \right \Vert_1 \nonumber \\ &\leq&  \sqrt{\frac{4 \ln(d) (18d)^{k} k^2 }{t}},
\end{eqnarray}
where we used the convexity of $x \mapsto x^2$. Define
\be
\pi := \Lambda_{Q_{j_1}} \otimes \ldots \otimes \Lambda_{Q_{j_t}} (\rho) = \sum_{a_1, \ldots, a_t} p(a_1, \ldots, a_t) \ket{a_1, \ldots, a_t} \bra{a_1, \ldots, a_t} \otimes \pi_{a_1, \ldots, a_t},
\ee
with $p(a_1, \ldots, a_t) := \tr(M_{a_1} \otimes \ldots M_{a_t} \rho_{Q_{j_1}\ldots, Q_{j_t}})$ and $\pi_{a_1, \ldots, a_t}$ the post-selected state on the remaining sites after measuring sites $Q_{j_1}\ldots, Q_{j_m}$ of $\rho$ and obtaining outcomes $a_1, \ldots, a_t$. Also define the fully separable state
\be
\sigma :=  \sum_{a_1, \ldots, a_t} p(a_1, \ldots, a_t) \ket{a_1, \ldots, a_t} \bra{a_1, \ldots, a_t} \otimes \pi_{a_1, \ldots, a_t}^{Q_{i_1}} \otimes \ldots \otimes \pi_{a_1, \ldots, a_t}^{Q_{i_l}},
\ee
with $\{i_1, \ldots, i_l \} = [n] \backslash \{  j_1, \ldots, j_m \}$. On one hand we have
\be
\tr \left( H  \pi  \right) \leq e_{0}(H) + \frac{t}{n}.
\ee
On the other hand by Eq. (\ref{bounddefinettinosymmetryapplied}),
\be
\mathop{\mathbb{E}}_{i_1\ldots, i_k}  \Vert \pi^{Q_{i_1}, \ldots, Q_{i_k}} - \sigma^{Q_{i_1}, \ldots, Q_{i_k}}  \Vert_1 \leq \sqrt{\frac{4 \ln(d) (18d)^{k} k^2 }{t}} + \frac{t}{n}
\ee

Thus 
\be
\tr(H \sigma) \leq \tr \left( H  \pi  \right)  + \sqrt{\frac{4 \ln(d) (18d)^{k} k^2 }{t}} \leq e_{0}(H) + 2\frac{t}{n} +  \sqrt{\frac{4 \ln(d) (18d)^{k} k^2 }{t}}.
\ee
Choosing $t = \sqrt{n}$ gives the result.

\end{proof}

\section{SDP hierarchies for the ground-state energy on graphs with low threshold rank}\label{sec:threshold-detail}

In this section, we analyze the use of the Lasserre hierarchy for
lower-bounding the ground-state energy of a Hamiltonian and 
prove \thmref{lowT}.

{\em Definition of the hierarchy:} 
Recall that $W_k$ is the subspace of operators with weight $\leq k$, defined in Eq. \eq{k-local-def}.
Given a Hamiltonian $H\in W_{k}$ (in
fact, we will consider here only $H\in W_2$),
we will minimize $M(H,I)$ over bilinear
functions $M:W_k \times W_k \rar \bbR$ satisfying
\begin{subequations}\label{eq:q-lasserre}
\ba
M(I,I) &= 1 & \label{eq:M-norm}\\
M(X,Y) &= M(X',Y')&\text{whenever } XY = X'Y'\label{eq:M-consistent}\\
M(X,X) &\geq 0 & \forall X \in W_k \label{eq:M-psd}
\ea\end{subequations}
Of course
$M$ can be represented as a matrix of dimension $\leq \binom{n}{k}
d^{2k}$, and the optimization over $M$ can be performed by
semidefinite programming, since it suffices to check Eq. \eq{M-consistent}
on a basis for $W_k$ and Eq. \eq{M-psd} is equivalent to the condition that
$M\succeq 0$.

To see that this is indeed a relaxation of the original problem, we observe that if $\rho \in
\cD((\bbC^d)^{\ot n})$ then $M(X,Y) := \tr (\rho XY)$ is a valid
solution.  Clearly it satisfies the normalization constraint
\eq{M-norm} and the constraint \eq{M-consistent}
that $M(X,Y)$ depends
only on $XY$.  For the PSD constraint \eq{M-psd}, observe that for any $X\in
W_k$ we have $M(X,X) = \tr (\rho X^2) \geq 0$, since $\rho, X^2$ are
both PSD.

Since $M(X,Y)$ depends only on $XY$, we abuse notation and write $M(XY) := M(X,Y)$.  Indeed, an equivalent formulation of the hierarchy is that we specify all reduced density matrices of $\leq 2k$ systems.  However, in this formulation, the condition that $M$ is PSD is less transparent. It will turn out that this global PSD constraint is crucial.  Classically a similar situation holds when this hierarchy (there known as the Lasserre hierarchy) is applied to combinatorial optimization problems.  In this case, dropping the global PSD condition results in a linear program (the Sherali-Adams relaxation) expressing only the condition that the marginal distributions be compatible.  (In the quantum case, semidefinite programming is still required even to properly define the marginals.)  It is known that this weaker form of optimization requires $k=\Omega(n)$ to give good approximations of many CSPs~\cite{CMM09,BGMT12}; and indeed such negative results are also known for the Lasserre hierarchy~\cite{Sch08}.  On the other hand, in \cite{BRS11} the Lasserre hierarchy was shown to give good approximations of 2-CSPs on graphs with low threshold-rank. These are defined as graphs for which the threshold rank $\rank_{\lambda}$, defined as the number of eigenvalues with value larger than $\lambda$, is bounded. 

\thmref{lowT} will follow by an adaptation of the techniques of \cite{BRS11}. In fact our algorithm for estimating the ground-state energy is essentially the same as the Propagation Sampling algorithm (Algorithm 5.5) of \cite{BRS11}.  The main new element is our introduction of state tomography.
 
Imagine that we measure each qudit with an informationally-complete POVM $\Lambda$ with $\tilde d := \poly(d)$ outcomes (cf. \lemref{info-complete}), corresponding to measurement operators $\{E_1,\ldots,E_{\tilde d}\}$.  The outcomes of these POVM yield exactly the same information as $M$, and thus constitute an equivalent optimization.  However, the outcomes of these measurements can be interpreted (using the language of \cite{BRS11}) as $2k$-local random variables with range $[\tilde d]$.  This terminology means that for any $i_1,\ldots,i_{2k}\in [n]$ we can define a probability distribution $p^{X_{i_1}\cdots X_{i_{2k}}}$ and that the marginals of these distributions are consistent for overlapping sets.  We let $X_i$ (for $i\in [n]$) denote the random variable corresponding to the measurement of qudit $i$\longonly{, and for each $x\in [\tilde d]$ let $X_{i,x} = 1_{\{X_i = x\}}$ be the indicator R.V. corresponding to a particular outcome of the measurement}.  \longonly{We will abuse notation and call $p$ a probability distribution with the understanding that only its marginals on $\leq 2k$ systems are actually defined.}  Since $\Lambda$ is informationally complete, $p$ carries exactly the same information as $M$ does.

\begin{mybox}
\begin{algorithm}[Quantum Propagation Sampling, based on Algorithm 5.5 of \cite{BRS11}]
\mbox{}\label{alg:propagation}
  \begin{description}
  \item[Input:] A matrix $M$ that is a valid solution to
    Eq. \eq{q-lasserre}.
  \item[Output:] A product state $\sigma = \sigma_1 \ot \cdots \ot \sigma_n$.
  \end{description}
  \begin{enumerate}
  \item Choose $m \in \{0,\ldots,2k-1\}$ at random.
  \item Choose $i_1, \ldots,i_m \in [n]$ randomly with replacement,
    and set $S = \{i_1,\ldots,i_m\}$.
  \item Sample an outcome $x_S=(x_{i_1},\ldots,x_{i_m})$ with
    probability $p^{X_S}(x_S)$. 
  \item \label{it:prop-penultimate}
Let $p|_{X_S=x_s}$ be the $2k-m$-local distribution resulting
    from conditioning on the outcome $X_S=x_S$.
  \item For every qudit $i\in [n]\setminus S$, set $\sigma_i$ to be the
    unique density matrix that fits the measurement outcomes
    $p^{X_i}|_{X_S=x_S}$; i.e. $\sigma_i = \Lambda^{-1}(p^{X_i}|_{X_S=x_S})$.
  \item For each $i\in S$, set $\sigma_i$ arbitrarily.
  \end{enumerate}
\end{algorithm}
\end{mybox}

\begin{repthm}{thm:lowT} 
\LowT
\end{repthm}

\begin{proof} 
We will prove that Algorithm~\ref{alg:propagation} achieves the error bound claimed by \thmref{lowT}.  First, we need to argue that step 5 always produces legal density matrices.  To see this, observe that $M$ can be used to define an $m+1$-partite density matrix $\rho^{S\cup \{i\}}$ satisfying $\tr(\rho^{S\cup \{i\}} X) = M(X^{Q_SQ_i})$ for any operator $X$ on $(\bbC^d)^{\ot m+1}$.  Then we see that $p^{X_S}(x_S) = \tr(\rho^{S\cup \{i\}}(E_{x_{i_1}}^{Q_{i_1}}\ot \ldots \ot E_{x_{i_m}}^{Q_{i_m}} \ot I))$ and $\sigma_i = \tr_S(\rho^{S\cup \{i\}}(E_{x_{i_1}}^{Q_{i_1}}\ot \ldots \ot E_{x_{i_m}}^{Q_{i_m}} \ot I)) / p^{X_S}(x_S)$.  This is simply the state resulting from obtaining measurement outcomes $x_{i_1},\ldots,x_{i_m}$ on systems $i_1,\ldots,i_m$, and thus must be a valid density matrix.

Now we argue that $\tr(H\sigma)\approx M(H)$.   This will use the analysis
of a closely related algorithm in \cite{BRS11}.  Theorem 5.6 of
\cite{BRS11} states that if $p$ is the set of measurement outcomes
induced from $M$ and $q$ is the product distribution on measurement
outcomes produced by step~\ref{it:prop-penultimate} of
Algorithm~\ref{alg:propagation} (i.e. $q = \bigotimes_i
p^{X_i}|_{X_S=x_S}$), then under certain conditions, we have
\be \E_{(i,j)\sim G} \|p^{X_i X_j} - q^{X_i X_j}\|_1 \leq \delta.
\label{eq:lasserre-decoupled}\ee
Those conditions are that:
\ba
k & \geq C \frac{\tilde d}{\delta^4} \rank_{\Omega(\delta/\tilde d)^2}(G)
\label{eq:k-condition}\\
& (\Cov_{p|_{X_S=x_S}}(X_{iy}, X_{jz}))_{iy,jz}
\text{ is PSD for any }S, x_S \label{eq:cov-PSD-condition}
\ea
Satisfying \eq{k-condition} will simply require taking $k$ large
enough.  For Eq. \eq{cov-PSD-condition}, we consider an arbitrary vector
$v\in \bbR^{(n-m)\tilde d}$, define $Z= \sum_{i\in [n]\setminus
  S}\sum_{y\in [\tilde d]} v_{i,y} E_{y}^{X_i}$ and observe that
\ba
\sum_{i,j\in [n]\setminus S}\sum_{y,z\in [\tilde d]} 
&v_{i,y} v_{j,z} \Cov_{p|_{X_S=x_S}}(X_{iy},X_{jz}) 
\\&=
\sum_{i,j\in [n]\setminus S}\sum_{y,z\in [\tilde d]} 
 v_{i,y} v_{j,z} 
\frac{M(E_{x_S}^{X_S} E_y^{X_i} E_z^{X_j}) M(E_{x_S}^{X_S})-
M(E_{x_S}^{X_S} E_y^{X_i})  M(E_{x_S}^{X_S} E_z^{X_j})}
{M(E_{x_S}^{X_S})^2} \\
& = \frac{M(Z^2E_{x_S}^{X_S}) M(E_{x_S}^{X_S}) - M(ZE_{x_S}^{X_S})^2}
{M(E_{x_S}^{X_S})^2}
\\ & = M'(Z^2) - M'(Z)^2. \label{eq:M-prime}
\ea
To obtain this final simplification, we define a functional $M'$ that acts on
$2k-m$-local operators on systems $[n]\setminus S$ by
\be M'(Y) := \frac{M(E_{x_S}^{X_S}Y)}{M(E_{x_S}^{X_S})}. \ee
We claim that $M'$ satisfies \eq{q-lasserre}.  Eqs. \eq{M-norm} and \eq{M-consistent} are
straightforward to verify. For Eq. \eq{M-psd}, we use the fact that $E_x$
are all rank-1 projectors and $E_{x_S}^{X_S}$ and $Y$ act on disjoint systems
to obtain
\be M'(Y^2) = 
\frac{M((\sqrt{E_{x_S}^{X_S}}Y)^2)}{M(\sqrt{E_{x_S}^{X_S}}^2)} \geq 0. 
\ee
Finally we set $Y = Z - M'(Z)I$ in Eq. \eq{M-prime} to complete the proof
of Eq. \eq{cov-PSD-condition}.  This establishes Eq. \eq{lasserre-decoupled}.

To relate this to quantum states, define $\rho_{i,j}$ to be the two-qudit density matrix satisfying $M'(X^{Q_i Q_j}) = \tr(\rho_{i,j}X)$.  The objective value of the SDP is $\E_{(i,j)\sim E} \tr(\rho_{i,j} H_{i,j} ) \leq e_0(H)$, so we will complete our proof by arguing that $\rho_{i,j} \approx \sigma_i \ot \sigma_j$ whenever $i,j\not\in S$.  Observe that $p^{X_i X_j} = \Lambda^{\ot 2}(\rho_{i,j})$ (when $i,j\not\in S$) and $q^{X_i X_j} = \Lambda(\sigma_i) \ot \Lambda(\sigma_j)$.  Thus  $ \|\rho_{i,j} - \sigma_i \otimes \sigma_j\|_1\leq 18 d \|p^{X_i X_j} - q^{X_iX_j}\|_1$ and
\ba
\tr(H\sigma) - M(H) & \leq \E_{(i,j), i, j \notin S} \tr (H_{i,j} (\sigma_i \ot \sigma_j - \rho_{i,j})) + \frac{2m}{n} \\
& \leq \E_{(i,j),  i, j \notin S} \|\sigma_i \ot \sigma_j - \rho_{i,j}\|_1 + \frac{2m}{n}  \\
& \leq 18d \E_{(i,j),  i, j \notin S} \|q^{X_i X_j} - p^{X_iX_j}\|_1 + \frac{2m}{n}  \\
& \leq 18 d \delta + \frac{2m}{n} 
\ea
We complete the proof by setting $\delta = \eps/72d$ and noting that $m/n \leq 2k/n \leq  \varepsilon/4$ so that this step achieves error $\leq \eps/2$.

Finally, for the special case where $H_{i,j}$ is of the form $-\sum_x A_x \ot A_{\pi_{i,j}(x)}$, Ref. \cite{BRS11} shows that Eq. \eq{k-condition} can be replaced by the condition that $k\geq \poly(1/\eps)$.   In this case, we let $p$ be the distribution resulting from measuring each system according to $\{A_x\}$; i.e. $p(X_{i_1}=x_{i_1},\ldots, X_{i_m} = x_{i_m}) = M(A_{x_{i_1}}^{Q_{i_1}} \ot \cdots \ot A_{x_{i_m}}^{Q_{i_m}})$.  Then \cite{BRS11} shows not that $\|p^{X_i X_j} - q^{X_iX_j}\|_1$ is small, but instead that $\sum_x |p^{X_iX_j}(x, \pi_{i,j}(x)) - q^{X_iX_j}(x, \pi_{i,j}(x))|$ is small.  Since this is all we need to bound $\tr(H\sigma) - M(H)$, we obtain the desired bound.
\end{proof}

\section{Proofs of de Finetti theorems without symmetry}
\label{eq:dF-proofs}

\begin{repthm}{thm:deFinetticlassical}
\deFinettiC{second}
\end{repthm}

\begin{proof}
By the chain rule of mutual information (Eq. (\ref{eq:chainrule}) of Lemma \ref{lem:MI}):
\begin{equation}
I(X_{i_1}\ldots, X_{i_k} : X_{j_1}, \ldots, X_{j_{t}}) = I(X_{i_1}\ldots, X_{i_k} : X_{j_1}) + \ldots + I(X_{i_1}\ldots, X_{i_k} : X_{j_t} | X_{j_1} \ldots X_{j_{t-1}}).
\end{equation}
Taking the expectation over $i_1, \ldots, i_k, j_1, \ldots, j_t$:
\begin{eqnarray}
&& \mathop{\mathbb{E}}_{j_1,\ldots,j_{t}} \mathop{\mathbb{E}}_{i_1,\ldots,i_k} I(X_{i_1}\ldots, X_{i_k} : X_{j_1}, \ldots, X_{j_{t}}) \\\nonumber  &=& \mathop{\mathbb{E}}_{j_1,\ldots,j_{t}} \mathop{\mathbb{E}}_{i_1,\ldots,i_k} I(X_{i_1}\ldots, X_{i_k} : X_{j_1}) + \ldots + \mathop{\mathbb{E}}_{j_1,\ldots,j_{t}} \mathop{\mathbb{E}}_{i_1,\ldots,i_k}I(X_{i_1}\ldots, X_{i_k} : X_{j_t} | X_{j_1} \ldots X_{j_{t-1}}).
\end{eqnarray}
Thus there is a $m \leq t$ such that (relabeling $j_m$ to $i_{k+1}$)
\begin{eqnarray}
&& \mathop{\mathbb{E}}_{j_1,\ldots,j_{m-1}} \mathop{\mathbb{E}}_{i_1,\ldots, i_{k+1}} I(X_{i_1}\ldots, X_{i_k} : X_{i_{k+1}} | X_{j_1} \ldots X_{j_{m-1}}) \nonumber \\ &\leq& \frac{1}{t} \mathop{\mathbb{E}}_{j_1,\ldots,j_{t}} \mathop{\mathbb{E}}_{i_1,\ldots,i_k} I(X_{i_1}\ldots, X_{i_k} : X_{j_1}, \ldots, X_{j_{t}}) \nonumber \\ &\leq& \frac{k \log |X|}{t},
\end{eqnarray}
where the last inequality follows from Eq. (\ref{upperlimit}).

Finally, from Eqs. (\ref{eq:multitobi}) and (\ref{monotonicitymulti}): 
\begin{eqnarray}
&& \mathop{\mathbb{E}}_{j_1,\ldots,j_{m-1}} \mathop{\mathbb{E}}_{i_1,\ldots,i_{k+1}} I(X_{i_1} : \ldots : X_{i_k+1} | X_{j_1}\ldots, X_{j_{m-1}}) \nonumber \\
&=&  \mathop{\mathbb{E}}_{j_1,\ldots,j_{m-1}} \mathop{\mathbb{E}}_{i_1,\ldots,i_k} \left(  I(X_{i_1} : X_{i_2} | X_{j_1}\ldots, X_{j_{m-1}}) +  \ldots + I(X_{i_1}\ldots X_{i_{k}} : X_{i_{k+1}} | X_{j_1}\ldots, X_{j_{m-1}})  \right) \nonumber \\
&\leq& k \mathop{\mathbb{E}}_{j_1,\ldots,j_{m-1}} \mathop{\mathbb{E}}_{i_1,\ldots,i_{k+1}}  I(X_{i_1}\ldots X_{i_{k}} : X_{i_{k+1}} | X_{j_1}\ldots, X_{j_{m-1}})  \nonumber \\
&\leq& \frac{k^2\ln |X|}{t}.
\end{eqnarray}
The statement then follows by applying Pinsker's inequality (Eq. (\ref{eq:pinsker}) of Lemma \ref{lem:MI}) to the quantity $ I(X_{i_1} : \ldots : X_{i_k} | X_{j_1}\ldots, X_{j_{m-1}})$.
\end{proof}


\section{Open Problems}

There are several open problems related to the results of this paper. 

\begin{itemize}

\item We start listing open problems mentioned before in the paper: To prove conjecture \ref{con:weighted} in section 2.1.3, to improve the dependence on $k$ in Theorem 17 from $d^k$ to $\poly(k)$, to obtain a generalization of Theorem 17 for the local norm (see discussion in section 2.3), and to obtain a generalization of Theorem 6 for non-regular graphs.

\item An important open question is whether the quantum PCP conjecture holds true. We believe the results of this paper put into check the validity of the conjecture. At least they show that a quantum analogue of the PCP theorem in conjunction with parallel repetition of 2-CSPs (in order to increase the degree of the constraint graph, but without the need to amplify the gap as in Raz's theorem \cite{Raz98}; see the proof of \propref{CSPs}) is \textit{false}. Hence it emerges as a interesting problem whether there is a meaningful notion of parallel repetition for quantum CSPs. Note that parallel repetition of CSPs involves massive copying of the variables. This is the reason why we do not know how to generalize it to the quantum case, where the assignment is a highly entangled state and thus there is no clear notion of copying of the variables. 


\item There has been some recent activity in the complexity of the local Hamiltonian problem for commuting Hamiltonians \cite{Has12, BV05, Sch11, AE11, Has12b}. On one hand it is possible that the problem is always in $\NP$, even for estimating the energy to inverse polynomial accuracy. On the other hand, it is possible that it is $\QMA$-hard to obtain a constant error approximation to the ground-state energy, settling the quantum PCP conjecture in the affirmative. For 2-local commuting Hamiltonians Bravyi and Vyalyi proved that the problem is always in $\NP$ by $C^*$-algebraic techniques \cite{BV05}. Later these techniques were generalized to more general Hamiltonians in \cite{Sch11, AE11, Has12}. Can this approach be combined with the methods of this paper to give a disproof of the quantum PCP conjecture for commuting Hamiltonians? One challenge would be to develop similar techniques for $k$-local Hamitonians with $k \geq 3$, since this is the relevant case for commuting Hamiltonians \footnote{While for addressing the general quantum PCP conjecture it suffices to consider 2-local models (since one can apply perturbation theory to reduce the general case to 2-local Hamiltonians \cite{BDLT08}), this is not the true for commuting Hamiltonians, since we do not have perturbation-theory gadgets that preserve commutativity.}. 

\item Can we get in $\NP$ an approximation to the ground energy within small error of general $k$-local Hamitonians on very good small-set expander graphs? This would be a common generalization of Theorem 8 (which establishes that for $2$-local models) and Aharonov and Eldar result \cite{AE11}  for commuting $k$-local models.


\item We have mentioned that our work is in the spirit of the mean-field approximation, and can be seen as a generalization of mean-field to Hamiltonians where different edges carry different interactions.  Mean-field theory, however, is more than merely approximating states with product states.  For example, there is a well-understood theory of corrections to mean-field theory, in which a tensor power state is viewed as a vacuum and fluctuations are viewed as bosonic excitations (see e.g. \cite{RS07}).  It would be interesting to have a systematic method of computing corrections to our product-state approximation for high-degree graphs.

\item Finally it would be interesting to developed information-theoretic tools to analyse the performance of variational ansatz beyond product states. For example, what is the power of constant-depth circuits applied to product states (see \cite{FreedmanH13})? Can information-theoretical methods give new results in this direction?

\end{itemize}

\longonly{
\section*{Acknowledgements}

We thank Itai Arad for catching an important bug in the proof of \thmref{clustered} in a previous version and thank Dorit Aharonov, Boaz Barak, Ignacio Cirac, Sev Gharibian, Will Matthews, David Poulin, Or Sattath, Leonard Schulman, and David Steurer for many interesting discussions and useful comments. Much of this work was done while FGSLB was working at the Institute for Theoretical Physics in ETH Z\"urich and AWH was working in the Department of Computer Science at the University of Washington. FGSLB acknowledges support from EPSRC, the polish Ministry of Science and Higher Education Grant No. IdP2011 000361, the Swiss National Science Foundation, via the National Centre of Competence in Research QSIT, the German Science Foundation (grant CH 843/2-1), the Swiss National Science Foundation (grants PP00P2$\textunderscore$128455, 20CH21$\textunderscore$138799 (CHIST-ERA project CQC)), the Swiss National Center of Competence in Research "Quantum Science and Technology (QSIT)", and the Swiss State Secretariat for Education and Research supporting COST action MP1006. AWH was funded by NSF grants CCF-0916400 and CCF-1111382 and ARO contract W911NF-12-1-0486.  
}

\end{document}